\documentclass[12pt,reqno]{amsart}
\usepackage{amsmath,amssymb,amsfonts,amsthm}
\usepackage[mathscr]{eucal}
\usepackage[all]{xy}
\label{}
\usepackage[english]{babel}

\usepackage{hyperref}
\title[Local BRST cohomology in (non-)Lagrangian field theory]
{Local BRST cohomology in (non-)Lagrangian field theory}

\author{D.S. Kaparulin,  S.L.  Lyakhovich and A.A. Sharapov}

\address{Department of Quantum Field Theory, Tomsk State University, Lenin ave. 36, Tomsk 634050, Russia.}

\email{dsc@phys.tsu.ru, sll@phys.tsu.ru, sharapov@phys.tsu.ru}

\newtheorem{prop}{Proposition}[section]

\newtheorem{definition}[prop]{Definition}
\newtheorem{theorem}[prop]{Theorem}
\newtheorem{corollary}[prop]{Corollary}
\theoremstyle{remark}

\renewcommand{\simeq}{\cong}

\def\gh{\mathrm{gh}}
\def\vf{\varphi}
\def\deg{\mathrm{deg}}
\def\Deg{\mathrm{Deg}}
\def\pgh{\mathrm{pgh}}

\begin{document}

\maketitle

\begin{flushright}
\emph{Dedicated to the 70th birthday of Igor Victorovich Tyutin}
\end{flushright}

\begin{abstract}
Some general theorems are established on the local BRST cohomology
for not necessarily Lagrangian gauge theories.  Particular
attention is given to the BRST groups with direct physical
interpretation.  Among other things, the groups of rigid
symmetries and conservation laws are shown to be still connected,
though less tightly than in the Lagrangian theory. The connection
is provided by the elements of another local BRST cohomology group
whose elements are identified with Lagrange structures. This
extends the cohomological formulation of the Noether theorem
beyond the scope of Lagrangian dynamics. We show that each
integrable Lagrange structure gives rise to a Lie bracket in the
space of conservation laws, which generalizes the Dickey bracket
of conserved currents known in Lagrangian field theory. We study
the issues of existence and uniqueness of the local BRST complex
associated with a given set of field equations endowed with a
compatible Lagrange structure. Contrary to the usual BV formalism,
such a complex does not always exist for non-Lagrangian dynamics,
and when exists it is by no means unique. The ambiguity and
obstructions are controlled  by certain cohomology classes, which
are all explicitly identified.

\end{abstract}

\section{Introduction}

The BRST approach provides  the most systematic method, sometimes
having no alternatives, for quantizing gauge systems.  The BRST
theory also gives a deep insight into the classical dynamics. In
particular, the classical BRST complex allows one to identify the
cohomological obstructions to switching on consistent interactions
in field theories and provides a cohomological understanding for
Noether's correspondence between symmetries and conservation laws.

The basic ingredient of the BRST theory is a homological vector
field $Q$, called the classical BRST differential, which is to be
associated, in one way or another, to the original classical
system. The practices of casting classical dynamics into the BRST
framework vary from \textit{ad hoc} constructions applicable to
the models of special types to uniform schemes that work equally
well for any Lagrangian or Hamiltonian system. Lagrangian dynamics
can always be put into the BRST framework by means of the BV
formalism \cite{BV}; for any (constrained) Hamiltonian system, the
BRST differential can be constructed by the BFV method \cite{BFV}.
Though both the methods have enjoyed enormous attention in various
reviews\footnote{For a systematic introduction to the field, we
refer to the book \cite{HT}.}, we would like to make here some
general remarks concerning these schemes of the BRST embedding.
The remarks address the key inputs and outputs of the BV and BFV
formalisms that are subject to a revision when the BRST formalism
is extended beyond the scope of Lagrangian/Hamiltonian dynamics.

The BV scheme of the BRST embedding starts with an action
functional, which is supposed to exist for the original classical
equations of motion. The classical BRST differential $Q$, being
constructed by the BV method, has a special feature: it is a
Hamiltonian vector field with respect to the BV antibracket
(canonical odd Poisson bracket), i.e., $Q=(S, \,\cdot\,)$. The BV
master action $S$ is constructed by solving the BV master equation
$(S, S)=0$ in the field-antifield space. The latter space is
constructed by adding ghosts and antifields to the original field
content. The original action functional provides the boundary
condition for the master action in the extended space.  Notice
that the BV prescription for the field-antifield space
construction essentially relies on the fact that the original
equations are Lagrangian, in which case the gauge symmetry
generators coincide with the generators of Noether identities
(Noether's second theorem). Beyond the scope of Lagrangian
dynamics, the latter property is not true anymore and one can find
many counterexamples.  As a result, not just the boundary
condition for the master equation is problematic to define, the BV
construction of the field-antifield space does not apply to a
system defined by non-Lagrangian equations. So, the canonical
antibracket cannot be a general tool for constructing the BRST
differential beyond the class of Lagrangian systems.

Somewhat similar scenario of the BRST embedding is played in the
BFV formalism. The basic input there is the (first class
constrained) Hamiltonian form of classical dynamics. The BRST
differential has the Hamiltonian form $Q=\{\Omega,\, \cdot \, \}$
on the ghost extension of the original phase space. The BRST
charge $\Omega$ is sought from the BFV master equation $\{\Omega,
\Omega \}=0$. The boundary condition for $\Omega$ is provided by
the first class constraints of the original system. The ghost
extension of the original phase space appeals to the two facts:
(i) the original phase space has been already equipped with a
Poisson bracket such that the evolutionary equations are
Hamiltonian, and (ii) there is a pairing between the first class
constraints and gauge symmetry generators. For a general system of
(constrained) evolutionary equations either assumption may be
invalid, so that the BFV framework cannot be directly applied to
general dynamical systems.

It should be noted that in classical theory, the Hamiltonian
structure of the BRST differential $Q$ is not always relevant.
Many of the deliverables of classical BRST complex appeal to its
properness (see Definition \ref{properness} in the next section)
rather than to the Hamiltonian form of the BRST differential. The
BFV/BV \textit{quantization}, however, essentially employs the
(anti)bracket, which is `a must' ingredient of the quantum BRST
operator.

In the papers \cite{LS}, \cite{KazLS}, \cite{LS1}, \cite{LS2}, a
general scheme has been worked out for the BRST embedding and
quantization of dynamical systems defined by general equations of
motion, not necessarily Lagrangian or Hamiltonian. In \cite{LS}, a
ghost extension of the original configuration space of fields has
been proposed in such a way as to implement the equations of
motion through the cohomology of a proper BRST differential $Q$.
The differential is iteratively constructed by solving the
equation $Q^2=0$ with the usual tools of homological perturbation
theory \cite{HT}, \cite{FHST}, \cite{FH}. In the irreducible case,
the boundary condition for $Q$ is provided by the classical
equations of motion together with the generators of their gauge
symmetries and Noether identities, with no pairing assumed between
the two groups of generators. This construction needs no bracket
to define the BRST differential, and it is sufficient to develop
many of the standard applications of the BRST theory to classical
dynamics.

In \cite{LS}, \cite{KazLS}, two extra structures were proposed to
define a consistent quantization of non-Hamiltonian/non-Lagrangian
dynamics. These are called respectively the weak Poisson structure
\cite{LS} and the Lagrange structure \cite{KazLS}. From the pure
algebraic viewpoint, either structure extends the classical BRST
differential $Q$, respectively, to $P_\infty$- or
$S_\infty$-algebra \footnote{Pure algebraic definitions of
$P_\infty$ and $S_\infty$ can be found in \cite{V}.}, where $Q$ is
taken to be the first structure map. The second structure map
defines a bracket (even for $P_\infty$ and odd for $S_\infty$)
which is compatible with the BRST differential and satisfies the
Jacobi identity up to a homotopy correction governed by $Q$. The
usual BV and BFV formalisms fit in this algebraic picture as very
special cases where the structure maps $P_k$ and $S_k$ vanish for
all $k>2$ and the brackets associated to $P_2$ and $S_2$ are
non-degenerate.

It should be noted that the existence of a compatible  Lagrange
structure is a much more relaxed condition for the classical
equations of motion than the requirement to be derivable from the
action principle. Many examples are known of non-Lagrangian field
theories admitting nontrivial Lagrange structures that allow one
to construct a reasonable quantum theory \cite{KazLS}, \cite{LS1},
\cite{LS2}, \cite{LS3}, \cite{KLS}, \cite{BG}, \cite{KLS1}. In all
the examples, the Lagrange structure enjoys the property of
space-time locality.

The concept of  space-time locality is of a paramount importance
for quantum field theory. For a long time, the locality of the
BRST differential remained a plausible hypothesis. It has been
proven later \cite{H} that this hypothesis is actually
superfluous: no obstructions can appear to the existence of a
local BV master action or a BRST charge provided the original
classical dynamics are local and regular. Further studies of the
local BRST cohomology revealed its remarkable properties together
with many important applications to field-theoretical problems,
see \cite{BBH} for review. In particular, imposing the locality
condition gives birth to many interesting groups of the BRST
cohomology (e.g. those describing the rigid symmetries and
conservation laws), which otherwise vanish identically.

As the BRST theory is now available for not necessarily Lagrangian
dynamics, it becomes an issue to study the local BRST cohomology
for this more general class of systems.  It is the topic we
consider in this paper. The addressed issues include all the main
problems concerning the local BRST cohomology solved earlier in
the Lagrangian setting, some of which are posed in a slightly
different way, and some others result in different conclusions.
The main results can be summarized as follows:

\begin{itemize}

\item A local classical BRST differential can always be defined
for any local gauge theory, be it Lagrangian or not. Incorporation
of a nontrivial Lagrange structure in the BRST complex is
generally obstructed by some classes of the local BRST cohomology.
We identify these classes and show that their number is finite.
Thus, unlike the usual BRST theory for Lagrangian gauge systems,
the space-time locality of non-Lagrangian equations of motion and
a compatible Lagrange structure is not generally sufficient for
the existence of a local BRST complex.

\item Certain vanishing theorems for and general relations between
various local BRST cohomology groups are established for general
non-Lagrangian gauge theories. Some of these groups admit
immediate physical interpretations linking them to the spaces of
characteristics (conservation laws), rigid symmetries, and
Lagrange structures.

\item A cohomological version of Noether's first theorem,
connecting conservation laws to rigid symmetries, is formulated
and proven for not necessarily Lagrangian gauge theories endowed
with a Lagrange structure. For general non-Lagrangian theories,
the rigid symmetries and conserved currents are not so tightly
related to each other as in the Lagrangian case, and a particular
form of this relation strongly depends on the choice of Lagrange
structure.

\item We extend the Dickey bracket of conserved currents to the
case of non-Lagrangian gauge theories equipped with integrable
Lagrange structures. Furthermore, the Lagrange structure is shown
to provide a homomorphism of the Dickey algebra of conservation
laws to the Lie algebra of rigid symmetries. We also introduce the
higher derived brackets that allow one to generate new
characteristics, symmetries, and Lagrange structures from
previously known characteristics.

\end{itemize}

The paper is organized as follows. In the next section, we briefly
review the definition of the BRST complex that serves any gauge
theory, be it Lagrangian or not. A special emphasis is placed on
space-time locality, including the definitions of local $p$-forms
and functionals of fields the BRST differential acts upon. Section 3
provides some auxiliary facts from homological algebra we use to
study various local BRST cohomology groups. In Section 4, we prove
several general theorems on the local BRST cohomology. The use of
spectral sequence arguments considerably shortens the proofs as
compared to the previously known Lagrangian counterparts. In Section
5, we elaborate on physical interpretation of certain groups of
local BRST cohomology. These are the groups that are naturally
identified with the spaces of rigid symmetries, conservation laws,
and Lagrange structures. Here we also elucidate the structure of the
BRST differential, more precisely its Koszul-Tate part, from the
viewpoint of the underlying gauge dynamics.  In Section 6, we
discuss some natural operations one can define for the local BRST
cohomology, the most important of which is the Dickey bracket on the
space of conservation laws. Section 7 deals with the issues of
existence and uniqueness for the local BRST complex. Contrary to the
usual BV formalism, the existence of local BRST complex is generally
obstructed by certain classes of the local BRST cohomology, which we
explicitly identify. In the last Section 8, we summarize our
results. Appendix A contains some technical details taken out of
Sections 2 and 6.

\section{A non-Lagrangian BRST complex}\label{1}

This section provides a glossary on the BRST theory of
non-Lagrangian gauge systems with special emphasis on space-time
locality. For a systematic account of the theory as well as its
application to the path-integral quantization of various
non-Lagrangian gauge models the reader is referred  to the
original papers \cite{KazLS}, \cite{LS1}, \cite{LS2}, \cite{LS3},
\cite{KLS1}.

In the local setting, a \textit{proper gauge system of type}
$(n,m)$ is defined by the following data.

\vspace{3mm}\noindent{\textbf{Source}}: A smooth manifold $X$ of
dimension $\Delta$ endowed with a suitable volume form $v$. We
will write $\{x^\mu\}_{\mu=1}^\Delta$ for a system of local
coordinates on $X$ and denote by $\Lambda=\bigoplus_p\Lambda^p(X)$
the exterior algebra of differential forms on $X$.

\vspace{3mm}\noindent\textbf{Target}: The cotangent bundle $T^\ast
M$ of a graded supermanifold endowed with the canonical symplectic
structure. The local coordinates on the base $M$ and the linear
coordinates on the fibers are denoted respectively by
\begin{equation}\label{phi-barphi}
\vf^I=(\vf^{i_k},\vf_{i_l})\,,\qquad \bar\vf_I=(\bar\vf_{i_k},
\bar\vf{}^{i_l})\,,\qquad k=0,\ldots,m\,,\quad l=1,\ldots,n+1\,,
\end{equation}
so that the canonical Poisson bracket on $T^\ast M$ is given by
\begin{equation}
\{\vf^I,\bar\vf_J\}=\delta_J^I\,.
\end{equation}
We will also use the collective notation $\phi^A=(\vf^I,
\bar\vf_J)$ for the whole set of local coordinates on $T^\ast M$.
The various gradings prescribed to the coordinates
(\ref{phi-barphi}) are collected  in Table 1.

\vspace{5mm}
\begin{center}
\begin{tabular}{|l|l|l|l|}
  \hline
    ghost number & momentum degree & resolution degree & pure ghost number \\
  \hline
   $\gh \,\vf^{i_k}=k$ &  $ \Deg \,\varphi^{i_k}=0$  &$ \deg \,\vf^{i_k}=0$  &$ \pgh \,\vf^{i_k}=k$   \\
  $\gh\,\bar\vf{}_{i_k}=-k $  &$  \Deg\, \bar\vf_{i_k}=1$    & $\deg\,\bar\vf_{i_k}=k+1 $& $\pgh \,\bar\vf_{i_k}=0$  \\
   $\gh\,\vf_{i_l}=-l$ & $ \Deg\, \vf_{i_l}=0$ & $\deg\,\vf_{i_l}=l$ &$ \pgh\,\vf_{i_l}=0$  \\
   $\gh\, \bar\vf{}^{i_l}=l$ & $\Deg\,\bar\vf{}^{i_l}=1  $&$ \deg\, \bar\vf{}^{i_l}=0$ & $ \pgh \,\bar\vf{}^{i_l}=l-1$\\
  \hline
\end{tabular}

\vspace{5mm} \textsc{Table} 1. The gradings of local coordinates
on the target space.
\end{center}

\vspace{5mm} To avoid cumbersome sign factors we assume that the
Grassmann parity of the coordinates is compatible with the ghost
number, i.e., the even coordinates have even ghost numbers and the
odd coordinates have odd ghost numbers.  In physical terms this
means that we consider gauge theories without fermionic degrees of
freedom. The results are easily extended to the general case of a
gauge theory with both bosonic and fermionic fields.

The four gradings of local coordinates are not actually
independent. Comparing the columns of Table 1, one can see that
\begin{equation}\label{N}
\pgh+\Deg=\gh+\deg\,.
\end{equation}

\vspace{3mm}\noindent\textbf{Fields}: Smooth maps from $X$ to
$T^\ast M$. The fields can be considered as points of an (infinite
dimensional) manifold $\mathcal{M}$. The local coordinate systems
on $\mathcal{M}$ are identified with the local coordinate
expressions of maps $\phi: X\rightarrow T^\ast M$. The manifold
$\mathcal{M}$ inherits the gradings and the symplectic structure
of $T^\ast M$. Namely, if $\phi: X\rightarrow T^\ast M$ is given
locally by $\phi^A(x)=(\vf^J(x)$, $\bar \vf_J(x))$, then the
symplectic 2-form on $\mathcal{M}$ reads
\begin{equation}\label{symstr}
\omega=\int_Xv(\delta\bar\vf_J\wedge\delta\vf^J)\,.
\end{equation}
Here the wedge denotes  the exterior product of variational
differentials.

From the field-theoretical viewpoint, the ``canonical momenta''
$\bar\vf_I(x)$ play the role of sources to the ``genuine fields''
$\vf^I(x)$ on the space-time manifold $X$.
  If $X$ is a manifold with boundary,
some boundary conditions on $\phi$'s should be imposed. For our
purposes it is sufficient to suppose all the sources vanishing on
the boundary together with all their derivatives,
\begin{equation}\label{bcon}
    \partial_{\mu_1}\cdots\partial_{\mu_n}\bar\vf_{I}|_{\partial
    X}=0\,,\qquad n=0,1,2,\ldots\,.
\end{equation}

In this paper, we are interested in gauge fields whose dynamics is
governed by local equations of motion. The standard way to
incorporate locality is to use the formalism of jet spaces
\cite{Olv}. Let us regard $\mathcal{M}$ as the space of sections
of the trivial fiber bundle $E=X\times T^\ast M$ over $X$ and  let
$J^{\infty}E$ denote the infinite jet bundle of $E$ over $X$. Each
smooth section of $E$, that is, a field $\phi\in \mathcal{M}$,
induces the section $j^{\infty}\phi$ of $J^\infty E$. By a {local
function} on $J^\infty E$ we mean the pullback of a smooth
function on some finite jet bundle $J^pE$ with respect to the
canonical projection $J^\infty E\rightarrow J^pE$. Similarly, by a
\textit{local function of fields} $\phi$ on $X$ we will mean the
pullback of a local function on $J^\infty E$ via the section
$j^\infty\phi: X \rightarrow J^\infty E$. The notion of local
function on $X$ can be further extended to the notion of a
\textit{local $k$-form} on $X$ by considering the differential
forms on $X$ with coefficients that are local functions of fields
$\phi$. In terms of local coordinates a local $k$-form on $X$
reads
\begin{equation}\label{}
\omega=\omega(x, \phi(x),\partial_\mu \phi(x),\ldots,
\partial_{\mu_1}\cdots\partial_{\mu_p}\phi(x))_{\mu_1\cdots
\mu_k}dx^{\mu_1}\wedge\cdots\wedge dx^{\mu_k}\,.
\end{equation}
The local forms constitute the exterior differential algebra
$\mathcal{A}=\bigoplus \mathcal{A}^{g,k}_{m,n}$ graded by the
ghost number $g$, momentum degree $m$, pure ghost number $n$, and
form degree $k$. The volume form $v$ on $X$ defines the natural
isomorphism $f_v: \mathcal{A}^{g,0}_{m,n}\rightarrow
\mathcal{A}^{g,\Delta}_{m,n}$. Finally, we define the
\textit{local functionals} to be the integrals over $X$ of local
$\Delta$-forms; in doing so, two local functionals are considered
to be  equivalent if their integrands differ by the exterior
differential of a local $(\Delta-1)$-form. For the functionals  of
positive momentum degree this equivalence relation, having  the
form of equality, is implicit in the boundary conditions
(\ref{bcon}). Denoting the graded vector space of local
functionals by $\mathcal{F}=\bigoplus \mathcal{F}_{m,n}^g$, one
can think of the integration over $X$ as a linear map
$$
\begin{array}{c}
\int_X :
\mathcal{A}^{g,\Delta}_{m,n}/d\mathcal{A}^{g,\Delta-1}_{m,n}\;\rightarrow\;
\mathcal{F}_{m,n}^g\,.
\end{array}
$$
This map is actually an isomorphism of vector spaces, see Theorem
\ref{APL} below.

The Poisson bracket associated to the symplectic structure
(\ref{symstr}) equips $\mathcal{F}$ with the structure of a graded
Lie algebra. Furthermore,  the bracket defines the Hamiltonian
action of $\mathcal{F}$ on the space of local functions. This
action then naturally extends to the action on the whole algebra
of local forms, so that we can think of $\mathcal{A}$ as a graded
module over the Lie algebra $\mathcal{F}$.

 \vspace{3mm}\noindent\textbf{BRST charge}: An odd local functional $\Omega$ of ghost
 number 1 satisfying the following two conditions:
 \begin{equation}\label{Omega}
 \mathrm{Deg}\,\Omega>0 \,,\qquad\{\Omega,\Omega\}=0\,,
 \end{equation}
plus a \textit{properness condition} to be specified below (see
Definition \ref{properness}). The first condition in (\ref{Omega})
can also be written as $\Omega|_{\bar\vf_I=0}=0$ and the second
one is known as the \textit{master equation}.

Define the \textit{BRST differential} $s$ on local functions from
$\mathcal{A}$ by
\begin{equation}
sA=\{\Omega, A\}\,, \qquad \gh\,s=1\,.
\end{equation}
By the Jacobi identity for the Poisson brackets, the BRST
differential squares to zero, $s^2=0$, and  endows the algebra
$\mathcal{A}$ with the structure of increasing cochain complex
with respect to the ghost number. This complex is called the
\textit{BRST complex} of a gauge system. The action of $s$ extends
trivially  from local functions to the whole algebra of local
forms by setting $s (dx^\mu)=0$ and, by the universal coefficient
theorem \cite{Mac}, we have the isomorphism of cohomology groups
$$H^g(s, \mathcal{A}^{\bullet,k})\simeq
H^g(s,\mathcal{A}^{\bullet, 0})\otimes \Lambda^k(X)\,.$$

\begin{prop}\label{p1} The expansion of the BRST differential
according to the resolution degree  is given by
\begin{equation}\label{sps}
s=\delta+ \stackrel{_{(0)}}{s}+\stackrel{_{(1)}}{s}+\cdots\,,
\end{equation}
where
$$
\mathrm{deg}\, \delta=-1\,,\qquad \mathrm{deg}\,
\stackrel{_{(n)}}{s}=n\,,\qquad \delta^2=0\,.
$$
\end{prop}
The only nontrivial part of the assertion is that the expansion
(\ref{sps}) is bounded from below by resolution degree -1. The
proof is given in Appendix A. The operator $\delta$, called the
\textit{Koszul-Tate differential}, makes the algebra $\mathcal{A}$
into a decreasing cochain complex with respect to the resolution
degree. Denote by $H^{(r)}(\delta)$, $r\geq 0$, the corresponding
cohomology groups.

\begin{definition}\label{properness}
The BRST charge $\Omega$ is said to be proper if
$$
H^{(0)}(\delta)\neq 0\quad \mbox{and}\quad H^{(r)}(\delta)=0\quad
\forall r>0\,.
$$
\end{definition}
In all the following, only the proper BRST charges are considered.

Consider now the expansion of the BRST charge in powers of
sources. By (\ref{Omega}) we have
\begin{equation}\label{Om-exp}
\Omega= \Omega_1+\Omega_2+ \Omega_3+\cdots \,, \qquad \mathrm{Deg}
\, \Omega_n=n\,,
\end{equation}
and
\begin{equation}\label{ClBRST}
\{\Omega_1,\Omega_1\}=0\,,\qquad \{\Omega_1,\Omega_2\}=0\,,\qquad
\{\Omega_2,\Omega_2\}=-2\{\Omega_1,\Omega_3\}\,,\quad\ldots\,.
\end{equation}
The leading term $\Omega_1$ is called the \textit{classical BRST
charge}. It generates  the Hamiltonian vector field
$s_0=\{\Omega_1,\,\cdot\,\}$ called the \textit{classical BRST
differential}\footnote{Notice that the differential $s_0$ is
completely determined by its restriction onto the subalgebra of
local functions with zero momentum degree. This restriction is
also called the classical BRST differential \cite{KazLS}.}. By
virtue of the first relation in (\ref{ClBRST}), $s_0^2=0$. As for
the BRST differential $s$, the action of $s_0$ extends trivially
from the local functions to the local forms, giving the space
$\mathcal{A}$ the structure of cochain complex with respect to the
ghost number. The differential $s_0$ carries all the information
about the classical dynamics, hence the name. The higher order
terms in the expansion (\ref{Om-exp}) can be viewed as a
deformation of the classical BRST charge. The deformation becomes
crucial at the level of quantization. In this paper, however, our
main concern will be in classical aspects of the BRST theory under
consideration.

\section{Some auxiliary facts and constructions from homological algebra}

\subsection{An exact sequence for relative cohomology groups}
Given a bicomplex $C=\bigoplus_{q,r} C^{q,r}$ with differentials
\begin{equation}\label{dd}
    d: C^{q,r}\rightarrow C^{q+1,r}\quad\mbox{and}\quad \delta
    : C^{q,r}\rightarrow C^{q,r-1}\,,
\end{equation}
one can define various cohomology  groups $H(d)$, $H(\delta)$,
$H(d|\delta)$, and $H(\delta|d)$. The definition of the first two
groups is standard \cite{Mac}, while the last two groups describe
the relative cohomology of $d$ modulo $\delta$ and $\delta$ modulo
$d$, respectively. More precisely, $H(d|\delta)$ is the ordinary
cohomology group of the quotient complex $C/\delta C$ with the
differential induced by $d$ and $H(\delta|d)$ describes the
cohomology of $\delta$ in the quotient  $C/dC$.

\begin{theorem}[\cite{DVHTV}]
If $H^{q,r}(d)=0$ for some value of $(q,r)$, there exists an exact
sequence
\begin{equation}\label{esd}
0 \rightarrow H^{q,r+1}(d|\delta)\rightarrow
H^{q,r+1}(\delta|d)\rightarrow H^{q+1,r+1}(\delta) \rightarrow
H^{q+1,r+1}(\delta|d) \rightarrow H^{q,r}(d|\delta)\rightarrow 0\,.
\end{equation}
\end{theorem}

\begin{corollary}\label{Cor1}
If in the exact sequence above the middle group
$H^{q+1,r+1}(\delta)$ is trivial, then
$$
H^{q,r+1}(d|\delta)\simeq H^{q,r+1}(\delta|d)\quad \mbox{and}
\quad H^{q+1,r+1}(\delta|d) \simeq H^{q,r}(d|\delta)\,.
$$
\end{corollary}

\begin{corollary}\label{Cor2}
If
$H^{q,r}(d)=H^{q+1,r+1}(\delta)=H^{q+1,r}(d)=H^{q+2,r+1}(\delta)=0$,
then
$$
H^{q+1,r+1}(d|\delta)\simeq H^{q,r}(d|\delta)\,.
$$
\end{corollary}

\begin{corollary}\label{Cor3}
If
$H^{q,r}(d)=H^{q+1,r+1}(\delta)=H^{q,r-1}(d)=H^{q+1,r}(\delta)=0$,
then
$$
H^{q+1,r+1}(\delta|d) \simeq H^{q,r}(\delta|d)\,.
$$
\end{corollary}

\subsection{Spectral sequence of filtered complex}
A decreasing filtration of a complex  $(C,d)$ of vector spaces is
given by a family  of subspaces $F^pC^k$ (one for each $k$)
satisfying the following conditions:
$$
\cdots \supset F^{p-1}C^k\supset F^pC^k\supset
F^{p+1}C^k\supset\cdots\,,
$$
$$
\bigcap_p F^pC^k=0\,,\qquad C^k=\bigcup_pF^pC^k\,,\qquad d
(F^pC^k)\subset F^p C^{k+1}\,.
$$
To each filtered complex one can associate a sequence of bigraded
differential vector spaces  $\{E_r, d_r\}_{r=0}^\infty$, called
the spectral sequence, with the property that
$E^{p,q}_0=F^pC^{p+q}/F^{p+1}C^{p+q}$ and $E_{r+1}=H(d_r)$. The
differentials $d_r: E_r^{p,q}\rightarrow E^{p+r,q-r+1}_r$ are
induced in a certain natural way by $d$. In particular, $d_0$ is
given by the standard differential in the quotient complex
$F^{p}C/F^{p+1}C$. The spectral sequence machinery is designed to
compute the cohomology of filtered complexes. Namely, under
certain regularity conditions, one can prove that
$H^{k}(d)\simeq\bigoplus_{p+q=k}E^{p+q}_\infty$. (Intuitively, the
terms of the spectral sequence  $\{E_r\}$ can be regarded as
successive approximations \textit{from above} to $E_\infty$.)
Every so often, the sequence of bigraded vector spaces $\{E_r\}$
stabilizes for small values of $r$, i.e., $E_{r}\simeq
E_{r+1}\simeq\cdots \simeq E_{\infty}$, in which case one says
that the spectral sequence ``collapses'' at the $r$th step. For a
systematic exposition of spectral sequences we refer the reader to
\cite{Mac}.

\subsection{$L_\infty$-algebra}\label{S-inf}  An \textit{$L_\infty$-algebra}\footnote{We follow here the sign conventions of \cite{V}.
A physicist-oriented discussion of the $L_\infty$-algebras can be
found in \cite{LSt}. }
 is a graded vector space  $V$ endowed with
multibrackets $L_n \in Hom_1(S^nV, V)$, $n\in \mathbb{N}$,
satisfying the generalized Jacobi identities
\begin{equation}\label{GJI}
\sum_{k+l=n}\sum_{(k,l)-\mathrm{shuffle}} (-1)^\epsilon L_{l+1}
(L_k(a_{\sigma
(1)},...,a_{\sigma(k)}),a_{\sigma(k+1)},...,a_{\sigma(k+l)})=0\,,\qquad\forall
n\in \mathbb{N}\,,
\end{equation}
where a $(k,l)$-shuffle is a permutation of indices $1,2,...,k+l$
such that $\sigma(1)<\cdots < \sigma(k)$ and
$\sigma(k+1)<\cdots<\sigma(k+l)$, while $(-1)^\epsilon $ is the
natural sign prescribed by the sign rule for permutation of
homogeneous elements $a_1,...,a_n\in V$.

By definition, $L_0$ is just a distinguished element of $V$.  An
$L_\infty$-algebra with $L_0=0$ is called \textit{flat}. In the
flat case, the generalized Jacobi identities for $n=1,2,3$ can be
written as
\begin{equation}\label{wantbr}
\begin{array}{c}
d^2a=0\,,\\[4mm]
    d(a,b)+(da,b)+(-1)^{\varepsilon(a)\varepsilon(b)}(db,a)=0\,,\\[4mm]
    ((a,b),c)+(-1)^{\varepsilon(b)\varepsilon(c)}((a,c),b)+(-1)^{\varepsilon(a)(\varepsilon(c)+\varepsilon(b))}((b,c),a)\\[4mm]
    +dL_3(a,b,c)+L_3(da,b,c)+(-1)^{\varepsilon(a)\varepsilon(b)}L_{3}(db,a,c)+(-1)^{(\varepsilon(a)+\varepsilon(b))\varepsilon(c)}L_3(dc,a,b)=0\,,
    \end{array}
\end{equation}
where we set ${d}a=L_1(a)$ and $(a,b)=L_2(a,b)$. As is seen the
unary bracket $L_1$ defines a  coboundary operator on $V$, which
is also a derivation of the binary bracket. The binary bracket, in
its turn, satisfies the Jacobi identity with the homotopy
correction governed by $L_3$. The higher Jacobi identities impose
a coherent set of restrictions on $L_3$ and higher homotopies.
Notice that $L_2$ induces the Lie algebra structure on
$L_1$-cohomology. The usual Lie algebras can be viewed as
$L_\infty$-algebras with $L_k=0$, $\forall k\neq 2$, and $L_2$
being given by the Lie bracket.

\section{General theorems on local BRST cohomology}

In Section \ref{1}, we have introduced, besides the BRST
differential $s$, two more differentials: the Koszul-Tate
differential $\delta$ and the classical BRST differential $s_0$.
Both of the differentials commute with the exterior differential
$d$ giving the algebra $\mathcal{A}$ two different bicomplex
structures. As a result we have various ``absolute''   and
relative cohomology groups
\begin{equation}\label{HHHH}
    H^{g}_m(d)^k_n\,, \quad H^g_m(\delta)^k_n\,,\quad H^g_m(s_0)^k\,, \quad
    H^g_m(\delta|d)^k_n
    \,, \quad H^g_m(d|\delta)^k_n\,, \quad H^g_m(s_0|d)^k\,.
\end{equation}
The differentials $d$, $\delta$, and $s_0$ being homogenous, the
elements of the groups above are represented by local forms with
definite ghost number $g$, momentum degree $m$, and form degree
$k$. The $d$- and $\delta$-cohomology groups are also graded by
the pure ghost number $n$. Sometimes it will be convenient to
label the cohomology groups by the resolution degree in place of
the ghost number. To distinguish between these two gradings we
enclose the resolution degree $r$ in round brackets. Equation
(\ref{N}) allows one to switch easily from one notation to
another:
$$
H^{(r)}_m(\ldots)^k_n=H^{m-r+n}_m(\ldots)^k_n\,,\qquad
H^g_m(\ldots)^k_n=H^{(m-g+n)}_m(\ldots)^k_n\,.
$$

The aim of this section is to formulate some general theorems
about the cohomology groups (\ref{HHHH}). Most of the theorems
below are almost identical to those of the usual BV theory of
Lagrangian gauge systems \cite{BBH}. This is not surprising, as
there is no great difference between Lagrangian and non-Lagrangian
theories at the level of the Koszul-Tate differential. One may
wonder why we are interested in the (relative) cohomology of the
operators $\delta$ and $s_0$, leaving  aside the cohomology of the
``parent'' BRST differential $s$. The reason is twofold.  On the
one hand,  certain cohomology groups of (\ref{HHHH}) have useful
interpretations in classical field theory, as we will see in Sec.
\ref{interp}, and  on the other hand the physically relevant
cohomology of the operator $s$ is defined in the space of
\textit{nonlocal} functionals and appears to be crucial only upon
quantization.

In the rest of the paper we assume the source manifold $X$ to be
(diffeomorphic to) a contractible domain in $\mathbb{R}^{\Delta}$.
Under this assumption the following statement, called sometimes
the ``algebraic Poincar\'e lemma'', is true.

\begin{theorem}\label{APL}The cohomology of $d$ in $\mathcal{A}$ is given  by
$$
H(d)^k\simeq \left\{%
\begin{array}{ll}
    \mathbb{R}, & \hbox{for $k=0$;} \\
    0, & \hbox{for $0<k<\Delta$;} \\
    \mbox{the space of local functionals $\mathcal{F}$}, & \hbox{for $k=\Delta$.} \\
\end{array}%
\right.
$$
\end{theorem}

\begin{proof} The proof can be found in many places, see e.g.
\cite{BBH}, \cite{Olv}. The classes of $H(d)^0$ are represented by
constant functions on $X$.
\end{proof}

\begin{theorem}\label{3.2a}
For all $r\geq0$, $k<\Delta$, and $m+n+r+k>0$  there are
isomorphisms
$$
\begin{array}{ll}
H^{(r+1)}_m(\delta|d)^k_n\simeq
H^{(r+1)}_m(d|\delta)^k_n\,,&\qquad
H^{(r+1)}_m(\delta|d)^{k+1}_n \simeq H^{(r)}_m(d|\delta)^k_n\,,\\[3mm]
H^{(r+1)}_m(d|\delta)^{k}_n \simeq
H^{(r)}_m(d|\delta)^{k-1}_n\,,&\qquad
H^{(r+2)}_m(\delta|d)^{k+1}_n \simeq H^{(r+1)}_m(\delta|d)^k_n\,.
\end{array}
$$
\end{theorem}
\begin{proof}
Fixing the numbers $m$ and $n$, consider the bicomplex (\ref{dd})
with $C^{k,m-g+n}=\mathcal{A}^{g,k}_{m,n}$. Then the isomorphisms
follow immediately from Corollaries \ref{Cor1} - \ref{Cor3} (the
conditions $r\geq0$, $k<\Delta$, and $m+n+r+k>0$  ensure acyclicity
of $d$ and $\delta$).
\end{proof}

\begin{theorem}\label{3.2}
There are isomorphisms
\begin{equation}\label{2iso}
H^{(1)}_0(d|\delta)^0_0=0\,,\qquad H^{(1)}_0(\delta|d)^{1}_0\simeq
H^{(0)}_0(d|\delta)^0_0/\mathbb{R}\,.
\end{equation}
\end{theorem}
\begin{proof}
If $[a]\in H^{(1)}_0(d|\delta)^0_0$, then $da=\delta b$ for some
$b$. Hence, $d\delta a=0$ and by Theorem \ref{APL}, $\delta a
=c\in \mathbb{R}$. If $c=0$, then $a=\delta m$ because $\delta$ is
acyclic in positive resolution degree. In that case $a$ represents
the zero class of $H^{(1)}_0(d|\delta)^0_0$. Therefore,
$H^{(1)}_0(d|\delta)^0_0\neq 0$ iff there exists a relative
$d$-cocycle $a$ with $\delta a=c\neq 0$. In that case, any
$\delta$-cocycle $n$ in resolution degree zero is trivial for we
can write it as $n=\delta(na/c)$. This contradicts to
nontriviality of the group $H^{(0)}(\delta)$.

To prove the second isomorphism in (\ref{2iso}) consider the exact
sequence
\begin{equation}\label{esdelta}
0 \rightarrow H^{k-1,r}(\delta|d)\rightarrow
H^{k-1,r}(d|\delta)\rightarrow H^{k-1,r-1}(d) \rightarrow
H^{k-1,r-1}(d|\delta) \rightarrow H^{k,r}(\delta|d)\rightarrow
0\,.
\end{equation}
It is obtained from (\ref{esd}) by interchanging the roles of $d$
and $\delta$. Setting $k=r=1$ and $m=n=0$, we get
$H^{0,1}(d|\delta)=0$ and $H^{0,0}(d)=\mathbb{R}$. Then
(\ref{esdelta}) reduces to the short exact sequence
\begin{equation}\label{H11-H00}
0\rightarrow \mathbb{R} \rightarrow H^{0,0}(d|\delta) \rightarrow
H^{1,1}(\delta|d)\rightarrow 0\,,
\end{equation}
from which the desired isomorphism follows.

\end{proof}

\begin{theorem}\label{T32} There are isomorphisms
$$
\begin{array}{lr}
H_{m}^{(r)}(\delta|d)^\Delta_n\simeq
H^{(r-1)}_m(\delta|d)^{\Delta-1}_n\simeq \ldots\simeq
H_{m}^{(1)}(\delta|d)^{\Delta-r+1}_n\simeq
H^{(0)}_m(d|\delta)^{\Delta-r}_n\,,&\quad r<\Delta\,;\\[3mm]

H_{m}^{(\Delta)}(\delta|d)^\Delta_n\simeq
H^{(\Delta-1)}_m(\delta|d)^{\Delta-1}_n\simeq \ldots\simeq
H_{m}^{(1)}(\delta|d)^{1}_n\simeq
H^{(0)}_m(d|\delta)^{0}_n\,,&\quad m+n>0\,;\\[3mm]

H_{0}^{(\Delta)}(\delta|d)^\Delta_0\simeq
H^{(\Delta-1)}_0(\delta|d)^{\Delta-1}_0\simeq \ldots\simeq
H_{0}^{(1)}(\delta|d)^1_0\simeq
H^{(0)}_0(d|\delta)^{0}_0/\mathbb{R}\,;&\\[3mm]
H_{m}^{(r)}(\delta|d)^\Delta_n\simeq
H^{(r-1)}_m(\delta|d)^{\Delta-1}_n\simeq \ldots\simeq
H_{m}^{(r-\Delta+1)}(\delta|d)^1_n\simeq
H^{(r-\Delta)}_m(\delta)^{0}_n= 0\,,&\quad r>\Delta\,.
\end{array}
$$
\end{theorem}
\begin{proof}
All but two isomorphisms follow  directly from  Theorem
\ref{3.2a}. The rightmost isomorphism of the third line is
established by Theorem \ref{3.2}. The other isomorphism, which is
not covered by Theorem \ref{3.2a}, is the last isomorphism
$H^{(2)}_0(\delta|d)^1_0\simeq H^{(1)}_{0}(\delta)^0_0=0$ of the
fourth line. By Theorem \ref{3.2a}, we have
$H^{(2)}_0(\delta|d)^1_0\simeq H^{(1)}_0(d|\delta)^0_0$ and the
latter group vanishes  due to Theorem \ref{3.2}.

\end{proof}

\begin{theorem}\label{T2}
If  $n>0$,  then $H^{(r)}_m(d|\delta)^k_n=0$ for all $k<\Delta$.
\end{theorem}

\begin{proof} Following the terminology of \cite{DVHTV}, we
will refer to the fields with strictly positive pure ghost number
as \textit{foreground fields}, treating all the other fields as
background ones.  The local forms from $\mathcal{A}$ are
polynomial in the foreground fields and their derivatives.
Assuming  $n>0$, define the filtered complex $F^{p+1}K^k\subset
F^{p}K^k$, where the subspace $F^pK^k$ consists of all local
$k$-forms $a\in \mathcal{A}^{g,k}_{m,n}$ containing no more than
$q=k-p$ derivatives of the foreground fields. Clearly, the
exterior differential $d: K^k\rightarrow K^{k+1}$ preserves the
filtration. Now consider the quotient complex $C=K/\delta K$ with
differential induced by $d$. The vector spaces in question,
$H^g_m(d|\delta)^k_n$, are just the cohomology groups of the
complex $C$. The decreasing filtration on $K$ induces a decreasing
filtration on $C$. By definition, $F^pC^k$ is the space of
equivalence classes  $a+\delta K$  of $k$-forms with $a\in
F^pK^k$. Let $\{E_r,d_r\}$ be the spectral sequence associated to
the filtered complex $C$.  The space $E^{p,q}_0$ spans the
equivalence classes $a+\delta K$ represented by local
$(p+q)$-forms $a$ involving exactly $q$-derivatives of the
foreground  fields. To describe the action of $d_0$ it is
convenient to decompose the exterior differential into the sum
 of two operators, $d={d}'+{d}''$, where the operator
${d}'$ affects only the foreground fields, while ${d}''$ deals
with the background fields as well as explicit  dependence of the
space-time coordinates. Then the class $a+\delta K\in E_0$ is a
$d_0$-cocycles iff $ {d}' a =\delta b $ for some $b\in K$. But the
differential ${d}'$ is known to be acyclic in all but top
form-degrees. Furthermore, the corresponding homotopy $h$,
connecting $d'$ to the projector onto the subspace of
$\Delta$-forms, can be chosen to be commuting with $\delta$. (For
the explicit construction of such a homotopy operator see
\cite{DVHTV}.) The last fact implies that $h$ passes through the
quotient $K/\delta K$ giving rise to a homotopy for $d_0$. As a
consequence, the term $E_1^{p,q}$ appears to be supported on the
antidiagonal $p+q=\Delta$, so that the spectral sequence must of
necessity collapse after the first step. Hence
$H^{(r)}_m(d|\delta)_n^k\simeq \bigoplus_{p+q=k}
E_{\infty}^{p,q}=0$ for all $k<\Delta$.
\end{proof}

\begin{theorem}\label{T35}
If $n>0$ and $r>0$, then $H^{(r)}(\delta|d)_n=0$.
\end{theorem}

\begin{proof}
Theorem \ref{3.2a} establishes the isomorphisms
$$
H^{(r)}_m(\delta|d)^k_n\simeq\left\{%
\begin{array}{ll}
    H_m^{(r-1)}(d|\delta)^{k-1}_n, & \hbox{for $k>0$\,,} \\[3mm]
    H_m^{(r)}(d|\delta)_n^{k}, & \hbox{for $k<\Delta$}\,, \\
\end{array}%
\right.
$$
whenever the resolution degree is positive. By Theorem \ref{T2},
all these groups are trivial at positive pure ghost number.
\end{proof}

Theorems \ref{3.2} - \ref{T35} may be summarized by saying that
each of the relative cohomology groups $H(\delta|d)$ and
$H(d|\delta)$ is either zero or isomorphic to one of the following
groups:
\begin{equation}\label{Non-isomorhic-groups}
    H_0^{(0)}(d|\delta)^0_0/\mathbb{R}\,,\quad
    H_m^{(0)}(d|\delta)^{k}_0\,,\quad H^{(r)}_m(d|\delta)^\Delta_n\,,\quad H^{(0)}_m(\delta|d)^{q}_n,
\end{equation}
for $k=0,\ldots,\Delta-1$, $q=0,\ldots,\Delta$, $m\geq 0$, $n\geq
0$, and $r\geq 0$.

Let us now turn to the  cohomology of the classical BRST
differential. The general expansion  for $s_0$ with respect to the
resolution degree reads
\begin{equation}\label{longd}
    s_0=\delta+\gamma +\stackrel{_{(1)}}{s_0}+\cdots\,,\qquad \deg\,
\delta=-1\,,\quad \deg \,\gamma=0\,,\quad \deg
\stackrel{_{(r)}}{s}_0=r\,.
\end{equation}
The operator $\gamma$ is known as the \textit{longitudinal
differential} \cite{HT}. It implements the gauge invariance.
Expanding the identity $s^2_0=0$ according to the resolution
degree, we obtain the infinite sequence of relations
$$
\delta^2=0\,,\qquad [\delta,\gamma]=0\,,\qquad
\gamma^2=[\delta,\stackrel{_{(1)}}{s}_0 ]\,,\quad\ldots\,.
$$
The first three of these relations mean that the action of
$\gamma$ descends to the cohomology of $\delta$ inducing there a
coboundary operator. The corresponding complex is known as the
longitudinal complex.

\begin{theorem}\label{3.7}
The cohomology of the classical  BRST differential $s_0$ is given
by
\begin{equation}
H^{g}_{m} (s_0)\simeq
\left\{%
\begin{array}{ll}
    0, & \hbox{for $m>g$\,;} \\[3mm]
    H^g_m(\gamma, H^{(0)}(\delta)_{g-m}), & \hbox{for $m\leq g$\,.} \\
\end{array}%
\right.
\end{equation} The group $H^g_m(\gamma, H^{(0)}(\delta)_{g-m})$ describes
the cohomology of $\gamma$ in the cohomology of $\delta$.
\end{theorem}
\begin{proof}  In view of (\ref{N}) the expansion (\ref{longd}) is
simultaneously the expansion according to the pure ghost number,
$$
\pgh\, \delta=0\,,\qquad \pgh\, \gamma=1\,,\qquad
\pgh\stackrel{_{(r)}}{s}_0=r+1\,.
$$
Since $\pgh\, s_0\geq 0 $ we have the  filtered complex
$F^{p+1}C^k\subset F^pC^k$, where  $F^pC^k=\bigoplus_{n=p}^\infty
\mathcal{A}^{k,\bullet}_{m, n}$. Define the spectral sequence
$\{E_r, d_r\}$ associated to this filtration. It is clear that
$E_0^{p,q}=\mathcal{A}^{p+q,\bullet}_{m,p}$ and $d_0=\delta$. The
Koszul-Tate differential being acyclic in positive resolution
degree, $H^g_m(\delta)_n=0$ for $r=n-g+m>0$. Therefore all the
non-zero terms of  $E_1^{q,p}\simeq H^{q+p}_m(\delta)_p$  lie on
the horizontal line $q=m$, see Fig. \ref{SpSeq}$a$. The spectral
sequence necessarily collapses at the second step; in so doing,
the differential $d_1: E^{p,q}_1\rightarrow E_1^{p,q+1}$ is
induced by $\gamma$:
$$
E_1^{p,q}\simeq H_m^{p+q}(\delta)_p \ni [a] \quad \mapsto\quad
d_1[a]=[\gamma a]\in H^{p+q+1}_m(\delta)_{p+1}\simeq
E_1^{q,p+1}\,.
$$
Thus, $E^{p,q}_2\simeq H^{p+q}_m(\gamma, H^{(0)}(\delta)_p)=0$
unless $q=m$ and we finally get
$$
H^g_m(s_0)\simeq \bigoplus_{p+q=g}E_2^{p,q}\simeq H^g_m(\gamma,
H^{(0)}(\delta)_{g-m})\,.
$$
Of course,  $H^{(0)}(\delta)_n\equiv 0$ for $n<0$.
\end{proof}

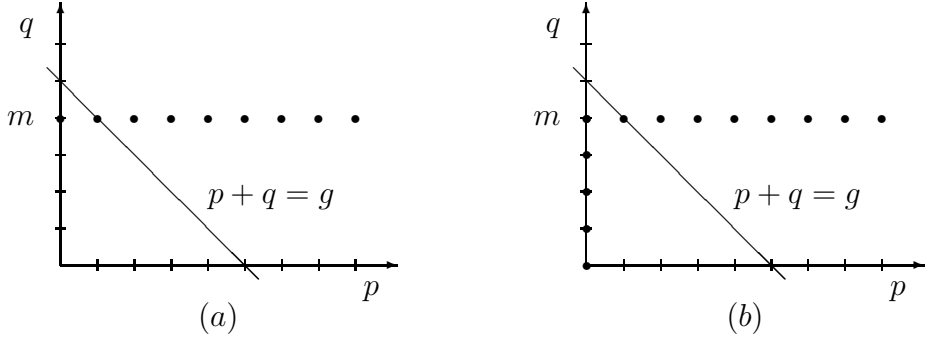
\begin{figure}
\unitlength 0.70mm 
\linethickness{0.4pt}\centering
\ifx\plotpoint\undefined\newsavebox{\plotpoint}\fi 
\begin{picture}(180.00,61.00)
\put(10.00,10.00){\vector(1,0){64}}
\put(17.00,11.00){\line(0,-1){2.00}}
\put(24.00,11.00){\line(0,-1){2.00}}
\put(31.00,11.00){\line(0,-1){2.00}}
\put(38.00,11.00){\line(0,-1){2.00}}
\put(45.00,11.00){\line(0,-1){2.00}}
\put(52.00,11.00){\line(0,-1){2.00}}
\put(59.00,11.00){\line(0,-1){2.00}}
\put(66.00,11.00){\line(0,-1){2.00}}
\put(5.00,38.00){\makebox(0,0)[rc]{$m$}}
\put(5.00,55.00){\makebox(0,0)[rc]{$q$}}
\put(10.00,10.00){\vector(0,1){50.00}}
\put(11.00,17.00){\line(-1,0){2.00}}
\put(11.00,24.00){\line(-1,0){2.00}}
\put(11.00,31.00){\line(-1,0){2.00}}
\put(11.00,38.00){\line(-1,0){2.00}}
\put(11.00,45.00){\line(-1,0){2.00}}
\put(11.00,52.00){\line(-1,0){2.00}}
\put(69.00,5.00){\makebox(0,0)[cc]{$p$}}
\put(40.00,0.00){\makebox(0,0)[cc]{$(a)$}}
\put(10.00,37.90){\circle*{1.5}} \put(17.00,37.90){\circle*{1.5}}
\put(24.00,37.90){\circle*{1.5}} \put(31.00,37.90){\circle*{1.5}}
\put(38.00,37.90){\circle*{1.5}} \put(45.00,37.90){\circle*{1.5}}
\put(52.00,37.90){\circle*{1.5}} \put(59.00,37.90){\circle*{1.5}}
\put(66.00,37.90){\circle*{1.5}}
\put(7.45,47.55){\line(1,-1){40.2}}
\put(110.00,10.00){\vector(1,0){64}}
\put(117.00,11.00){\line(0,-1){2.00}}
\put(124.00,11.00){\line(0,-1){2.00}}
\put(131.00,11.00){\line(0,-1){2.00}}
\put(138.00,11.00){\line(0,-1){2.00}}
\put(145.00,11.00){\line(0,-1){2.00}}
\put(152.00,11.00){\line(0,-1){2.00}}
\put(159.00,11.00){\line(0,-1){2.00}}
\put(166.00,11.00){\line(0,-1){2.00}}
\put(169.00,5.00){\makebox(0,0)[cc]{$p$}}
\put(140.00,0.00){\makebox(0,0)[cc]{$(b)$}}
\put(109.80,10.00){\vector(0,1){50.00}}
\put(110.80,17.00){\line(-1,0){2.00}}
\put(110.80,24.00){\line(-1,0){2.00}}
\put(110.80,31.00){\line(-1,0){2.00}}
\put(110.80,38.00){\line(-1,0){2.00}}
\put(110.80,45.00){\line(-1,0){2.00}}
\put(110.80,52.00){\line(-1,0){2.00}}
\put(105.00,38.00){\makebox(0,0)[rc]{$m$}}
\put(105.00,55.00){\makebox(0,0)[rc]{$q$}}
\put(110.00,37.90){\circle*{1.5}}
\put(117.00,37.90){\circle*{1.5}}
\put(124.00,37.90){\circle*{1.5}}
\put(131.00,37.90){\circle*{1.5}}
\put(138.00,37.90){\circle*{1.5}}
\put(145.00,37.90){\circle*{1.5}}
\put(152.00,37.90){\circle*{1.5}}
\put(159.00,37.90){\circle*{1.5}}
\put(166.00,37.90){\circle*{1.5}}
\put(110.00,37.90){\circle*{1.5}}
\put(110.00,30.90){\circle*{1.5}}
\put(110.00,23.90){\circle*{1.5}}
\put(110.00,16.90){\circle*{1.5}}
\put(110.00,09.90){\circle*{1.5}}
\put(50.00,23.00){\makebox(0,0)[cc]{$p+q=g$}}
\put(150.00,23.00){\makebox(0,0)[cc]{$p+q=g$}}
\put(107.45,47.55){\line(1,-1){40.2}}
\end{picture}
\caption[]{The first terms of spectral sequences from Theorems
\ref{3.7}, \ref{p}. }\label{SpSeq}
\end{figure}

\begin{theorem}\label{p}
There are isomorphisms of the cohomology groups
\begin{equation}\label{PhOb}
H^{g}_{m} (s_0|d)\simeq  \left\{%
\begin{array}{ll}
    H^{g}_{m} (\delta |d)_{0} , & \hbox{for $m>g$;} \\[3mm]
    H^g_m(\gamma,
    H(\delta|d)_{g-m}), & \hbox{for $m\leq g$,} \\
\end{array}%
\right.
\end{equation}
where  $H^p_m(\gamma,H(\delta|d)_{g-m})$ describes the cohomology
of $\gamma$ in the cohomology of $\delta$ modulo $d$.
\end{theorem}

\begin{proof}
We can proceed analogously to the proof of Theorem \ref{3.7}. For
a fixed $m$, consider the filtered complex $F^{p+1}C^k\subset
F^{p}C^k$, where $F^pC^k=\bigoplus_{n=p}^\infty
\mathcal{A}^{k,\bullet}_{m, n}/d\mathcal{A}^{k,\bullet-1}_{m, n}$.
By definition,  the space $H^k_m(s_0|d)$ is the $k$-th cohomology
group of the complex $C$. Associated to the filtration above is
the spectral sequence $\{E_r,d_r\}$ with
$E_0^{q,p}=\mathcal{A}^{p+q,\bullet}_{m,p}/d\mathcal{A}^{p+q,\bullet-1}_{m,p}$.
The differential $d_0$ is naturally induced by $\delta$ and we can
identify $E_1^{p,q}$ with $H_m^{p+q}(\delta|d)_p$. By Theorem
\ref{T35}, $H_m^{p+q}(\delta|d)_p=0$ whenever  $p>0$ and $q<m$.
Nonnegativity of the resolution degree $r=p-(p+q)+m$ also implies
that $q\leq m$. All potentially non-zero spaces of $E_1$ are
depicted in Fig. \ref{SpSeq}$b$. As is seen, they are nested
either on the horizontal line $q=m$ or on the vertical segment
$p=0$, $0\leq q<m$. By dimensional reasons, the spectral sequence
collapses from $E_2$. Notice that the differential $d_1$ is
induced by $\gamma$ and becomes zero upon restriction to the
spaces $E_1^{p,q}$ with $q<m$. Therefore, $E_2^{p,q}\simeq
E_1^{p,q}$ for $q<m$ and we see that
$$
H^g_m(s_0|d)\simeq\bigoplus_{p+q=g}E^{p,q}_2\simeq H^g_m(\gamma,
H(\delta|d)_{g-m}) \quad \mbox{for}\quad g\geq m\
$$
and
$$
H_m^g(s_0|d)\simeq\bigoplus_{p+q=g}E^{p,q}_2\simeq\bigoplus_{p+q=g}E^{p,q}_1\simeq
H^g_m(\delta|d)_0\quad \mbox{for}\quad g<m\,.
$$
The proof is complete.
\end{proof}

The most notable among the groups (\ref{PhOb}) is
$H^0_0(s_0|d)^\Delta$. It is the group that is identified with the
space of physical observables of a gauge system \cite{KazLS}.

\section{Interpretation of the groups
$H^{g}_{g+1}(\delta|d)^\Delta_0$}\label{interp}

In this section, we consider some special groups of relative
$\delta$-cohomology with a straightforward physical
interpretation. Namely, we are interested in the groups
$H^{g}_{g+1}(\delta|d)^\Delta_0$ for small $g$'s. The elements of
these groups are represented by local functionals in pure ghost
number zero and resolution degree one. To streamline  our
notation,  we will write $\mathcal{L}_g=\mathcal{F}_{g+1,0}^{g}$.
Notice that the graded vector space
$$\mathcal{L}=\bigoplus_{g=-1}^\infty \mathcal{L}_g$$ is closed
with respect to the Poisson bracket on $\mathcal{F}$,
\begin{equation}\label{A}
    \{\mathcal{L}_n,\mathcal{L}_m\}\subset \mathcal{L}_{n+m}\,,
\end{equation}
so that we can speak of the graded Lie algebra  $\mathcal{L}$ of
depth $-1$.

To clarify the structure of this algebra it is convenient to use
the condensed notation introduced in \cite{KazLS}:
$$
\begin{array}{llll}
\phi^i=\varphi^{i_0}(x)\,,\qquad&
\bar\phi_i=\bar\varphi_{i_0}(x)\,,\qquad&
\eta_a=\varphi_{i_1}(x)\,,\qquad&
\bar\eta^a=\bar\varphi^{i_1}(x)\,,\\[3mm]
c^\alpha =\varphi^{i_1}(x)\,,\qquad& \bar
c_\alpha=\bar\varphi_{i_1}(x)\,,\qquad&
\xi_A=\varphi_{i_2}(x)\,,\qquad&\bar\xi^A=\bar\varphi^{i_2}(x)\,.
\end{array}
$$
The superindices on the left hand side comprise both the discrete
indices of the fields and the space-time coordinates $x^\mu$; in
so doing, a repeated superindex means usual summation over its
discrete part and integration over $X$ with the measure $v$. All
the partial derivatives are understood as variational ones.

With this notation the Koszul-Tate differential takes the form
\begin{equation}\label{KT}
    \delta=T_a(\phi)\frac{\partial}{\partial
    \eta_a}+\bar\eta^a\partial_i
    T_a(\phi)\frac{\partial}{\partial \bar\phi_i}+\eta_aZ^a_A(\phi)\frac{\partial}{\partial
    \xi_A}+\bar\phi_iR^i_\alpha(\phi)\frac{\partial}{\partial \bar
    c_\alpha}+\bar\eta^a\eta_bU^b_{\alpha
    a}(\phi)\frac{\partial}{\partial \bar c_\alpha}+\cdots\,.
\end{equation}
Here  the dots stand for the terms differentiating by fields of
resolution degree $>2$.

Now it is a good point to explain in which sense $\delta$ is a
Koszul-Tate differential. The ``Koszul part'' of
$\delta=\delta_K+\delta_T$ is given by the first two terms in
(\ref{KT}) and appears to be Hamiltonian:
\begin{equation}\label{K}
\delta_K=\{\stackrel{_{(0)}}{\Omega}, \,\cdot\,\}\,,
\end{equation}
with $\stackrel{_{(0)}}{\Omega}=\bar \eta^aT_a(\phi)$ being the
resolution-degree-zero part of the BRST charge $\Omega$.

 Let us
introduce the supercommutative subalgebra of local forms
$\mathcal{B}\subset \mathcal{A}$  defined by the conditions
$$
\deg\, \mathcal{B} =0\,,\qquad \pgh\, {\mathcal{B}}=0\,,
$$
and let $\mathcal{I}$ be an ideal of $\mathcal{B}$ generated by
the local functions $T_a(\phi)$ and $\bar\eta ^a\partial_i T_a$.
Then the Koszul-Tate deferential $\delta$ defines a homological
resolution of the quotient algebra $\mathcal{B}/\mathcal{I}$.
Namely, let $({\mathcal{A}}_0,\delta)$ be the subalgebra of the
differential  supercommutative algebra $(\mathcal{A}, \delta)$
constituted by the local forms of pure ghost number zero, $\pgh
\,{\mathcal{A}}_0=0$. The algebra $\mathcal{A}_0$ is
$\mathbb{N}$-graded by the resolution degree and contains
$\mathcal{B}$ as the subalgebra of degree $0$.  Since $\delta$ is
acyclic in positive resolution degree,
 we have
\begin{equation}\label{}
    \begin{array}{l}
      H_k({\mathcal{\mathcal{A}}_0},\delta)=0\,, \qquad k>0\,, \\[3mm]
      H_0({\mathcal{A}}_0,\delta)\simeq \mathcal{B}/\mathcal{I} \,.\\
    \end{array}
\end{equation}
From the physical viewpoint, the generators of the ideal
$\mathcal{I}$ can be regarded as equations of motion for the
fields $\phi^i$ and $\bar\eta{}^a$,
\begin{equation}\label{ExS}
T_a(\phi)=0\,,\qquad \bar\eta ^a\partial_i T_a(\phi)=0\,.
\end{equation}
 The first set of
equations describes the dynamics of the original gauge fields
$\phi$. The second group of equations, called adjoint, plays an
auxiliary role as the fields $\bar\eta$'s, being of ghost number
$1$, are unphysical. We will refer to  (\ref{ExS}) as the
\textit{extended system} of dynamical equations. By construction,
the group $H_0(\mathcal{A}_0,\delta)$ is naturally graded by the
ghost number and its ghost-number-zero subgroup is isomorphic to
the algebra of local forms of fields $\phi$'s modulo equivalence
relation: two forms are considered to be equivalent if they take
the same values on each solution to the equations of motion
$T_a(\phi)=0$. Denoting by $\mathcal{N}$ the space of fields
$\phi$ and by $\Sigma\subset \mathcal{N}$ the subspace of
solutions to the equations $T_a(\phi)=0$, we can say that the
Koszul-Tate differential implements the restriction of the local
forms of fields $\phi\in \mathcal{N}$ to the subspace
$\Sigma\subset \mathcal{N}$. In the physical literature the space
of solutions  $\Sigma$ to the classical equations of motion is
called  \textit{shell}.

Notice  that for the theories of type $(0,0)$ the Koszul-Tate
differential $\delta$ reduces to $\delta_K$. In this case, the
original equations of motion $T_a(\phi)=0$ are necessarily
independent and enjoy no gauge freedom. If the equations happen to
be dependent (reducible), then there is an (overcomplete, in
general) basis of Noether's identity generators $Z^a_A$ such that
\begin{equation}\label{NI}
Z_A^aT_a=0\,.
\end{equation}
On the other hand, the presence of gauge symmetries for the
equations $T_a=0$ implies the existence of an (overcomplete, in
general) basis of gauge symmetry generators
$R_\alpha=R_\alpha^i\partial_i$ together with structure functions
$U_{\alpha a}^b$ such that
\begin{equation}\label{GS}
R_\alpha ^i\partial_i T_a=U_{\alpha a}^bT_b\,.
\end{equation}
Unlike the Lagrangian case, where we can identify  $Z$'s with
$R$'s due to  Noether's second theorem, there is no duality
between the gauge symmetries and the Noether identities for
general non-Lagrangian dynamics. The duality, however, is restored
at the level of the extended system (\ref{ExS}). Each generator of
the original identities (\ref{NI}) is simultaneously a generator
of the gauge transformation in the extended space:
\begin{equation}\label{dZ}
\delta_{\varepsilon}\phi^i=0\,,\qquad \delta_\varepsilon
\bar\eta{}^a=\varepsilon^AZ^a_A(\phi)\,.
\end{equation}
Furthermore, every gauge transformation $\delta_\varepsilon
\phi^i=\varepsilon^\alpha R^i_\alpha(\phi)$ of the original
equations of motion gives rise to the Noether identity for the
extended system (\ref{ExS}). (It is obtained by contracting
(\ref{GS}) with $\bar\eta^a$.) This duality is an immediate
consequence of the fact that the extended system of equations
(\ref{ExS}) is variational. Although the corresponding ``action
functional'' $\stackrel{_{(0)}}{\Omega}=\bar\eta^aT_a$ is odd, the
reasoning of the second Noether's theorem still applies to this
situation. As a result, the generators of Noether identities for
(\ref{ExS}) coincide with the generators of gauge symmetries and
the same is true for the structure functions defining the
higher-order reducibility conditions (if any). All these structure
functions, including $Z$'s and $R$'s, are incorporated in Tate's
part $\delta_T$ of the Koszul-Tate differential (\ref{KT}).

Having explained the ``physical meaning'' of the Koszul-Tate
differential, let us come back to the graded Lie algebra
$\mathcal{L}$. The general element of the homogeneous subspace
$\mathcal{L}_g$ has the form
\begin{equation}\label{A}
 A=\bar\phi_iA^i_{a_1\cdots
a_{g}}(\phi)\bar\eta^{a_1}\cdots
\bar\eta^{a_{g}}+\eta_aA^a_{a_1\cdots
a_{g+1}}(\phi)\bar\eta^{a_1}\cdots \bar\eta^{a_{g+1}}\,.
\end{equation}
Observe that the restriction of the Koszul-Tate operator
(\ref{KT}) onto the subspace $\mathcal{L}$ is given by the Koszul
differential (\ref{K}). As a result,  the action of
$\delta|_\mathcal{L}=\delta_K$ is Hamiltonian and the kernel
$\mathcal{Z}=\ker \delta|_{\mathcal{L}}$ appears to be a
subalgebra in the graded Lie algebra $\mathcal{L}$,
 \begin{equation}
 \mathcal{Z}=\bigoplus_{g=-1}^\infty \mathcal{Z}_g\,,\qquad  \{\mathcal{Z}_n,
\mathcal{Z}_m\}\subset \mathcal{Z}_{n+m}\,.
\end{equation}
Actually, a more strong statement is true: The Lie algebra
structure on $\mathcal{Z}$ descends to the $\delta$-cohomology.
This means that all the $\delta$-coboundaries from $\mathcal{Z}$
form an ideal in the Lie algebra $\mathcal{Z}$, so that it makes
sense to speak of the  Lie algebra structure on the quotient space
$$
\mathcal{Z}/(\mathcal{Z}\cap \delta \mathcal{F})\simeq
\bigoplus_{g=-1}^\infty H^g_{g+1}(\delta|d)^{\Delta}_{0}\,.
$$
The last fact is not obvious at all as  the space
$\mathcal{Z}\cap\delta \mathcal{F}$ is not exhausted by  $\delta
\mathcal{L}$, while the action of $\delta$ on the whole of
$\mathcal{F}$ is non-Hamiltonian. A rigorous definition of the Lie
bracket on $\mathcal{Z}/(\mathcal{Z}\cap \delta \mathcal{F})$ will
be given in Sec. \ref{ms}.

The element (\ref{A}) of ${\mathcal{L}}_g$ is a $\delta$-cocycle
iff
\begin{equation}\label{ccond}
    \delta A = (\partial_i T^{}_{a_1}A_{a_2\cdots
    a_{g+1}}^i+T_aA^a_{a_1\cdots a_{g+1}})\bar\eta^{a_1}\cdots
    \bar\eta^{a_{g+1}}= 0\,.
\end{equation}
Since the ghost fields $\bar\eta$'s are all odd and independent,
the last condition implies that
\begin{equation}\label{ccond1}
\partial_i T^{}_{[a_{1}}A_{a_2\cdots
    a_{g+1}]}^i+T_aA^a_{a_1\cdots a_{g+1}}=0\,,
\end{equation}
where the square brackets mean  antisymmetrization in the usual
(i.e., non-graded) sense. The cocycle $A$ is a $\delta$-coboundary
if there is a local functional
\begin{equation}\label{}
\begin{array}{lll}
    B&=& \bar\phi_i\bar\phi_j B^{ij}_{a_1\cdots
    a_{g-1}}\bar\eta^{a_1}\cdots
    \bar\eta^{a_{g-1}}+\bar\phi_i\eta_a B^{ai}_{a_1\cdots
    a_{g}}\bar\eta^{a_1}\cdots \bar\eta^{a_{g}}+\eta_a\eta_b
    B^{ab}_{a_1\cdots a_{g+1}}\bar\eta^{a_1}\cdots \bar\eta^{a_{g+1}}\\[3mm]
    &+& \bar c_\alpha B^\alpha_{a_1\cdots a_{g}}
    \bar\eta^{a_1}\cdots \bar\eta^{a_{g}}+\xi_A B^A_{a_1\cdots
    a_{g+1}}\bar\eta^{a_1}\cdots \bar\eta^{a_{g+1}}
    \end{array}
\end{equation}
such that $A = \delta B$. Explicitly,
\begin{equation}\label{cob}
\begin{array}{l}
    A^i_{a_1\cdots a_{g}} =2\partial_jT^{}_{[a_1}B^{ij}_{a_2\cdots
    a_{g}]} + T_aB^{a i}_{a_1\cdots a_{g}}+R_\alpha^iB^\alpha_{a_1\cdots
    a_{g}}\,,\\[3mm]
A^a_{a_1\cdots a_{{g+1}}}=\partial_iT^{}_{[a_1}B^{ai}_{a_2\cdots
a_{g+1}]}-2T_bB^{ab}_{a_1\cdots a_{g+1}}-U_{\alpha [a_1}^a
B^\alpha_{a_2\cdots a_{g+1}]}+Z_A^aB^A_{a_1\cdots a_{g+1}}\,.
    \end{array}
\end{equation}

Consider now the condition (\ref{ccond}) for $g=-1,0,1,2$.

\subsection{
The space of characteristics is $H_0^{-1}(\delta|d)^\Delta_0$.}
For $g=-1$ the cocycle condition (\ref{ccond}) reduces
to\footnote{So far we have considered local functionals modulo
boundary terms, that is, integrals of total derivatives. The total
derivatives, however, cannot be ignored when discussing the
conservation laws; hence, we write them explicitly here.}
\begin{equation}\label{char1}
A^a T_a =\int_{X} d\omega
\end{equation}
for some local $(\Delta-1)$-form $\omega$. According to the
condensed notation, the l.h.s. of (\ref{char1}) is given by the
integral of a linear differential operator  acting on the
equations of motion. Whenever the result of such an action is an
exact $\Delta$-form, one refers to $A=\{A^a\}$ as the generator of
\textit{identities} for the equations of motion $T_a(\phi)=0$. The
trivial cocycles (\ref{cob}) correspond to linear combinations of
trivial and Noether's identities, namely,
$$
A^a=2T_bB^{ba}+Z_A^aB^{A}\,.
$$
Since $d\omega$ vanishes on the shell $\Sigma$, the
($\Delta-1$)-form $\omega$ gives rise to the conserved current
$j=\ast \omega$, where the Hodge operator
$\ast:\Lambda^p(X)\rightarrow \Lambda^{\Delta-p}(X)$ is defined by
an appropriate metric on $X$.

The quotient of the whole space of identities by the Noether and
trivial identities is known as the space of characteristics
$\mathrm{Char}(T)$ for the equations $T_a(\phi)=0$. This leads us
to the following identification:
\begin{equation}\label{char}
\mathrm{Char}(T)=H^{-1}_{0}(\delta|d) \,.
\end{equation}

Notice that Rel.(\ref{char1}) does not specify the current
$j=\ast\omega $ unambiguously, since one can add to $\omega$ any
closed (and hence exact) $(\Delta-2)$-form $d\alpha$ as well as
any local $(\Delta-1)$-form $\gamma$ that is proportional to the
equations of motion and their differential consequences. (The
latter redefinition can be absorbed by the l.h.s. of
(\ref{char1}).) Modding out by these ambiguities, we get a
well-defined class $[\omega]\in H_0^{0}(d|\delta)^{\Delta-1}_0 $
associated to a characteristic $[A]\in
H^{-1}_0(\delta|d)_0^\Delta$. The assignment $[A]\mapsto [\omega]$
defines the bijection between the groups
$H^{-1}_0(\delta|d)_0^\Delta$ and
$H^0_0(d|\delta)_0^{\Delta-1}/\delta_{\Delta,1}\mathbb{R}$ stated
by Theorem \ref{T32}. In physical terms, the latter group can be
identified with the space of nontrivial conservation laws
$\mathrm{CL}(T)$ for the equations of motion $T_a(\phi)=0$. Thus,
we arrive at the following isomorphism:
\begin{equation}\label{Chat=CL}
\mathrm{Char}(T)\simeq \mathrm{CL}(T)\,.
\end{equation}

\subsection{ The space of rigid symmetries is $H^{0}_{1}(\delta|d)^{\Delta}_{0}$}
For $g=0$ the cocycle condition (\ref{ccond1}) takes the  form
\begin{equation}
    A^{i}\partial_{i}T_{a}+A^{b}_{a}T_{b}=0\,.
\end{equation}
It means that the vector field $A=A^i\partial_i$ acting  on the
space $\mathcal{N}$ of fields $\phi$ defines  a symmetry of the
equation of motion. The set of all the symmetries
$\mathrm{Sym}'(T)$ form a Lie algebra with respect to the
commutator of vector fields. The trivial symmetries (i.e., vector
fields vanishing on the shell $\Sigma$) and the gauge symmetries
correspond to the coboundaries (\ref{cob}):
$$
A^i=B^\alpha R_\alpha^i+T_aB^{ai}\,,\qquad
A_a^b=\partial_iT_aB^{bi}-2T_cB_a^{cb}-U_{\alpha a}^bB^\alpha
+Z^b_AB^A_a\,.
$$
The space of rigid (or global) symmetries $\mathrm{RSym}(T)$ is
defined to be the quotient of all the symmetries by trivial and
gauge symmetries. This leads to the identification
\begin{equation}\label{sim}
    \mathrm{RSym}(T)=H^{0}_{1}(\delta|d)_0^\Delta\,.
\end{equation}
Notice that the homogeneous subspace $\mathcal{Z}_0\subset
\mathcal{Z}$ is also a  Lie algebra. If
$$
A^1=\bar\phi_iA_1^i+\eta_a A_{1b}^a\bar\eta^b\,,\qquad
A^2=\bar\phi_iA_2^i+\eta_a A_{2b}^a\bar\eta^b
$$
are two relative $\delta$-cocycles associated to symmetries
$A_1,A_2\in \mathrm{Sym}'(T)$, then the Poisson bracket
$\{A^1,A^2\}$ is again a relative $\delta$-cocycle  associated to
the commutator of the vector fields $[A_1,A_2]\in
\mathrm{Sym}'(T)$. Thus, we have  an isomorphism between the Lie
algebras $\mathcal{Z}_0$ and $\mathrm{Sym}'(T)$. As mentioned
above, this isomorphism descends to the  cohomology, inducing a
Lie algebra structure on the space of all rigid symmetries
(\ref{sim}).

\subsection{ The space of Lagrange structures is $H^{1}_{2}(\delta|d)^\Delta_0$} The elements of $H^{1}_{2}(\delta|d)^{\Delta}_0$
are represented by the local functionals
\begin{equation}\label{LS}
A=\bar\phi_iA^{i}_{a}(\phi)\bar{\eta}^{a}+\eta_cA_{ab}^{c}(\phi)\bar{\eta}^{a}\bar{\eta}^{b}\,,
\end{equation}
where the structure functions $A^i_a$ and $A_{ab}^c$ obey the
cocycle condition
\begin{equation}\label{LA}
    A^i_{a}\partial_{i}T_{b}-A^i_{b}\partial_{i}T_{c}+2A_{ab}^{c}T_{c}=0\,.
\end{equation}
The last equation  is nothing but the defining relation for a
\textit{Lagrange structure} \cite{KazLS}. The set of vector fields
$A_a=A_a^i\partial_i$ on  $\mathcal{N}$ is called the
\textit{Lagrange anchor}. The trivial cocycles (\ref{cob})
correspond to the trivial Lagrange structures with
\begin{equation}\label{TLA}
A_a^i=2\partial_jT_aB^{ji}+ T_bB^{bi}_a+R^i_\alpha
B^\alpha_a\,,\qquad A^c_{ab}=\partial_i
T_{[a}B_{b]}^{ci}-2T_dB_{ab}^{cd}-U^c_{\alpha[a}B_{b]}^\alpha+Z^c_AB^A_{ab}\,.
\end{equation}
The group $H^{1}_{2}(\delta|d)^\Delta_0$ is thus naturally
identified with the space of local Lagrange structures modulo
trivial ones,
\begin{equation}\label{LA-2}
    \mathrm{LS}(T) = H^{1}_{2}(\delta|d)^\Delta_0\,.
\end{equation}

We know that the Poisson square of the cocycle (\ref{LS})
representing a Lagrange structure $[A]$ is a relative
$\delta$-cocycle from $\mathcal{Z}_2$. This cocycle may well be
nontrivial.

\begin{definition}\label{int}
A Lagrange structure $[A]\in H^{1}_{2}(\delta|d)_{0}^\Delta$ is
said to be integrable if $$[\{A,A\}]=0\in
H^2_3(\delta|d)^\Delta_0\,.$$
\end{definition}

The reason for introducing the notion of integrability  is
twofold. For one thing, each BRST charge involves an integrable
Lagrange structure; for another, each integrable Lagrange
structure gives rise to a Lie bracket on the space of conservation
laws together with a Lie algebra homomorphism to the space of
rigid symmetries. Both of these statements will be detailed in the
next two sections. Definition \ref{int} suggests to interpret the
group $H^2_3(\delta|d)^\Delta_0$ as the \textit{space of
obstructions to integrability}.

Not so much experience has been gained yet of computing the group
$\mathrm{LS}(T)$ even for linear equations of motion. In the
recent paper \cite{BG1} it is shown that the study of local BRST
cohomology for not necessarily Lagrangian gauge theories can be
always reduced to the case with the classical BRST differential
involving only the first space-time derivatives of fields. This
first order reduction is achieved by introducing not necessarily
finite number of auxiliary fields. It is worth to mention in this
connection the work \cite{BG} devoted to the study of local BRST
cohomology in the AKSZ sigma-model \cite{AKSZ}. The point is that
the AKSZ sigma-model, being dynamically empty in the bulk, defines
classical dynamics and quantization of the boundary degrees of
freedom\footnote{It is appropriate to note that every field theory
in $d$ dimensions, be it Lagrangian or not, can be converted into
an equivalent topological Lagrangian model in $d+1$ dimensions
following the systematic procedure proposed in Ref. \cite{KazLS}.
Whenever the field equations in $d$ dimensional space (considered
as the boundary of $d+1$ dimensional bulk) have the form of free
differential algebra (FDA), the corresponding $d+1$ topological
Lagrangian theory has the form of the AKSZ sigma-model.}. In
particular, the corresponding topological action should involve
some Lagrange structure for the boundary dynamics. It was shown
\cite{BG} that whenever the AKSZ sigma-model has a \textit{finite
dimensional} target space, the group $\mathrm{LS}(T)$ appears to
be isomorphic to the space of bi-vectors on the target space. The
integrability condition (\ref{int}) amounts then to the Jacobi
identity for the corresponding bi-vector. In principle, any field
theory can be brought to the FDA form at the cost of introducing
an \textit{infinite dimensional} target space. In that case,
however, the group $\mathrm{LS}(T)$ is far from being exhausted by
the bi-vectors. What is more, the physically relevant Lagrange
structures do not reduce to the target space bi-vectors even for
the FDA form of the d'Alembert equation \cite{KLS1}.

\section{Multiplicative structures}\label{ms}

Theorem  \ref{p} establishes an isomorphism
\begin{equation}\label{pi}
\pi : H^g_m(s_0|d)\rightarrow H^g_m(\delta|d)_0\,,\qquad m>g\,,
\end{equation}
between the relative cohomology groups. A particular realization
of this isomorphism  is as follows. Take a relative $s_0$-cocycle
$a$ representing a class of $H^g_m(s_0|d)$ and expand it according
to the pure ghost number,
$$ a=a_0+a_1+\cdots\,,\qquad \pgh\, a_n=n\,.
$$ Then $a_0$ is a relative
$\delta$-cocycle representing a class of $H^g_m(\delta|d)_0$.
Conversely, any cohomology class $[a_0]\in H^g_{g+1}(\delta|d)_0$
is uniquely extended to the class $[a]\in H^g_{g+1}(s_0|d)$ by
adding terms of positive pure ghost number to the representative
cocycle $a_0$.

Below we apply the isomorphism (\ref{pi}) to define various
multiplicative structures on the group
$$H^{(1)}(\delta|d)^\Delta_0=\bigoplus_{g=-1}^\infty H^g_{g+1}(\delta|d)^\Delta_0\,,
$$
some of which have been already appeared in the previous section.

With the Hamiltonian action of the classical BRST differential
$s_0=\{\Omega_1,\,\cdot\,\}$, the groups $H^g_{g+1}(s_0|d)^\Delta$
form a graded Lie algebra with respect to the Poisson bracket: for
any $[a]\in H^g_{g+1}(\delta|d)^\Delta_0$ and $[b]\in
H^{g'}_{g'+1}(\delta|d)_0^\Delta$ we have
$$
\{[a],[b]\}=[\{a,b\}]\in H^{g+g'}_{g+g'+1}(s_0|d)^\Delta\,.
$$
The isomorphism (\ref{pi}) allows one to transfer this Lie algebra
structure to the group $H^{(1)}(\delta|d)^\Delta_0$ by setting
\begin{equation}\label{LB}
\{\alpha,\beta\}=\pi (\{\pi^{-1}(\alpha),
\pi^{-1}(\beta)\})\,,\qquad \forall \alpha,\beta \in
H^{(1)}(\delta|d)^\Delta_0\,.
\end{equation}
The validity of the Jacobi identity is obvious. Here we
deliberately denote the pulled back Lie bracket on
$H^{(1)}(\delta|d)^\Delta_0$ by braces. The reason is that the
right hand side of (\ref{LB}) coincides exactly with the
cohomology class of the Poisson bracket of relative
$\delta$-cocycles representing the classes $\alpha$ and $\beta$.
The proof of this fact is given in Appendix A. Thus, one can
multiply (representatives of)  the classes $\alpha$ and $\beta$ as
such, i.e.,  omitting the maps $\pi^{-1}$ and $\pi$. Formula
(\ref{LB}) just explains why the result of multiplication does not
depend on the choice of representative cocycles.

The Lie bracket on $H^{(1)}(\delta|d)^{\Delta}_0$ gives rise to a
rich variety of interesting multiplicative structures on the
spaces of characteristics (\ref{char}), rigid symmetries
(\ref{sim}) and Lagrange structures (\ref{LA-2}).   First of all,
the rigid symmetries, being supported in ghost number zero, form a
closed Lie algebra with respect to the Poisson bracket.  They also
act on the spaces of characteristics and Lagrange structures in
the Hamiltonian way:
$$
\begin{array}{ll}
\mathrm{RSym}(T)\times \mathrm{Char}(T)\rightarrow
\mathrm{Char}(T)\,:&\quad (X,\Psi)\mapsto \{X,\Psi\}\,,\\[3mm]
\mathrm{RSym}(T)\times \mathrm{LS}(T)\rightarrow
\mathrm{LS}(T)\,:&\quad (X,\Lambda)\mapsto \{X,\Lambda\}\,.
\end{array}
$$
So, we can regard the spaces $\mathrm{Char}(T)$ and
$\mathrm{LS}(T)$ as modules over the Lie algebra
$\mathrm{RSym}(T)$. Notice that unlike the Lagrange structures,
the characteristics form an abelian Lie algebra with respect to
the Poisson bracket.

Consider now the bilinear map
$$
\mathrm{LS}(T)\times \mathrm{Char}(T) \rightarrow
\mathrm{RSym}(T)\,: \qquad (\Lambda,\Psi)\mapsto
\{\Lambda,\Psi\}\,,
$$
which assigns to a Lagrange structure $\Lambda$ and a
characteristic $\Psi$ the rigid symmetry
$X_\Psi=\{\Lambda,\Psi\}$. The rigid symmetries belonging to the
image of this map are called \textit{characteristic symmetries}.
Fixing the Lagrange structure $\Lambda$, we obtain a linear map
$f_\Lambda=\{\Lambda,\,\cdot\,\}$ from the space of
characteristics to the space of rigid symmetries. The homomorphism
\begin{equation}\label{f}
f_\Lambda: \mathrm{Char}(T) \rightarrow \mathrm{RSym}(T)
\end{equation}
can be regarded as a systematic extension of the first Noether
theorem to non-Lagrangian equations of motion \footnote{For a
modern exposition of the Noether theorem, including historical
background and miscellaneous generalizations, we refer the reader
to the recent book \cite{K-S}.}. Indeed, on account of
(\ref{Chat=CL}), we deduce that any conservation law gives rise to
a characteristic symmetry. Moreover, for the Lagrangian systems
endowed with the \textit{canonical} Lagrange structure, the map
(\ref{f}) is injective and its image coincides with the rigid
symmetries that preserve the corresponding action functional
\cite{KLS}. In the general case, the homomorphism (\ref{f}) is
neither injective nor surjective.

The natural question arises whether it is possible to pull back
the Lie algebra structure to the space of characteristics  via the
homomorphism (\ref{f}). It turns out that such a Lie algebra
structure on $\mathrm{Char}(T)$ does exist whenever $\Lambda$ is
integrable in the sense of Definition \ref{int}. The corresponding
Lie bracket on characteristics can be defined as the derived
bracket
\begin{equation}\label{DB}
 (\Psi_1,\Psi_2)_\Lambda=\{\{\Lambda,\Psi_1\},\Psi_2\}\,, \qquad \forall
 \Psi_1,\Psi_2\in \mathrm{Char}(T)\,.
\end{equation}
Notice that
$$
\gh (\,\cdot\,,\,\cdot\,)_\Lambda=1\,,\qquad
\mathrm{Deg}(\,\cdot\,,\,\cdot\,)_{\Lambda}=0\,.
$$
The Jacobi identity for this bracket follows from two facts: (i)
the characteristics form an abelian algebra with respect to the
Poisson bracket and (ii) the Poisson square of the  Lagrange
structure is equal to zero, $\{\Lambda,\Lambda\}=0$. From the
Jacobi identity for the Poisson bracket it also follows that
$$
(\Psi_1,\Psi_2)_\Lambda=-(\Psi_2,\Psi_1)_\Lambda\,.
$$
Now one can verify that (\ref{f}) is a Lie algebra homomorphism.
Indeed, if $X_{\Psi_1}=\{\Lambda,\Psi_1\}$ and
$X_{\Psi_2}=\{\Lambda,\Psi_2\}$, then
$$
\begin{array}{lll}
\{X_{\Psi_1},
X_{\Psi_2}\}&=&\{\{\Lambda,\Psi_1\},\{\Lambda,\Psi_2\}\}=\{\Lambda,\{\{\Lambda,\Psi_1\},\Psi_2\}\}-\{\{\Lambda,\{\Lambda,\Psi_1\}\},\Psi_2\}\\[3mm]
 &=&\{\Lambda,
(\Psi_1,\Psi_2)_\Lambda\}-\frac12\{\{\{\Lambda,\Lambda\},\Psi_1\},\Psi_2\}=X_{(\Psi_1,\Psi_2)_\Lambda}\,.
\end{array}
$$

The Lie algebra structure on the space of conservation laws was
first introduced by Dickey in the context of Lagrangian field
theory without gauge symmetry \cite{D}. The extension to
Lagrangian gauge theories has been found by Barnich and Henneaux
\cite{BH}. On account of the isomorphism (\ref{Chat=CL}) one can
view the Lie bracket (\ref{DB}) as a further generalization of the
Dickey algebra to the case of not necessarily Lagrangian gauge
theories endowed with an integrable Lagrange structure.

It should be noted that the homomorphism
$$
(\,\cdot\,,\,\cdot\,)_\Lambda:\mathrm{Char}(T)\times
\mathrm{Char}(T)\rightarrow \mathrm{Char}(T)
$$
maps  characteristics to characteristic for any Lagrangian
structure $\Lambda$, be it integrable or not. Integrability is
only crucial for the bracket (\ref{DB}) to meet the Jacobi
identity. Leaving aside the Jacobi identity, one can regard
(\ref{DB}) as a special case of a more general construction that
associates a family of multi-brackets to any element $\Theta \in
H^{m-1}_m(\delta|d)^\Delta_0$. These multi-brackets are given by
$$
(\Psi_1,\ldots, \Psi_{m-k})^k_\Theta=\{\ldots \{\{\Theta,
\Psi_1\},\Psi_2\},\ldots, \Psi_{m-k}\}\,,\qquad 0\leq k< m\,.
$$
Setting $k=0,1,2$ we get the multi-linear maps
$$
\begin{array}{l}
(\cdots )_\Theta^0: \mathrm{Char}(T)^{\times m}\rightarrow \mathrm{Char}(T)\,,\\[3mm]
(\cdots)^1_\Theta: \mathrm{Char}(T)^{\times m-1}\rightarrow \mathrm{RSym}(T)\,,\\[3mm]
(\cdots)^2_\Theta: \mathrm{Char}(T)^{\times m-2}\rightarrow
\mathrm{LS}(T)\,.
\end{array}
$$
The last formulae provide a useful interpretation of the higher
cohomology groups $H^{m-1}_m(\delta|d)^\Delta_0$ as  multi-linear
operations on characteristics.

A somewhat different family  of multi-brackets, generalizing the
bracket (\ref{DB}), is generated by the BRST charge itself.
Consider the expansion (\ref{Om-exp}) of the BRST charge according
to the momentum degree. The second equation in (\ref{ClBRST}) says
that the term $\Omega_2$ is a $s_0$-cocycle. Hence, it defines a
class $[\Omega_2]\in H^{1}_2(s_0|d)^\Delta$ and, via the
isomorphism (\ref{pi}), a Lagrange structure $\pi([\Omega_2])\in
H^1_2(\delta|d)_0^\Delta$. Due to the third relation in
(\ref{ClBRST}), the Lagrange structure $\Lambda=\pi([\Omega_2])$
is integrable and defines the Lie bracket (\ref{DB}) on the space
of characteristics. One can put this Lie algebra structure on
$\mathrm{Char}(T)$ in a more wide context of $L_\infty$-algebras
(see Section \ref{S-inf}). In \cite{KazLS}, it was noticed that
any BRST charge associated to a gauge system gives rise to an
$L_{\infty}$-structure on the space $\mathcal{F}_{0,\bullet}$ of
local functionals in momentum degree zero.  The corresponding
multibrackets $L_n: \mathcal{F}_{0,\bullet}\rightarrow
\mathcal{F}_{0,\bullet}$ are defined as the higher derived
brackets \cite{V}:
\begin{equation}\label{Ln}
L_n(a_1,\ldots,a_n)=\{\ldots\{\{\Omega_n,a_1\},a_2\},\ldots,a_n
\}\,.
\end{equation}
Using the Jacobi identity for the Poisson bracket and
commutativity of the Poisson algebra $\mathcal{F}_{0,\bullet}$,
one can see that the the multibrackets (\ref{Ln}) are graded
symmetric and satisfy the generalized Jacobi identities
(\ref{GJI}) by virtue of the master equation $\{\Omega,
\Omega\}=0$. It then follows from Rels. (\ref{wantbr}) that the
binary bracket
$$
L_2(a_1,a_2)=\{\{\Omega_2,a_1\},a_1\}
$$
induces a Lie bracket in the cohomology of the coboundary operator
$L_1=\{\Omega_1,\,\cdot\,\}:\mathcal{F}_{0,\bullet}\rightarrow
\mathcal{F}_{0,\bullet}$. This yields a Lie algebra structure on
$H_0(s_0|d)^\Delta$. The space of characteristics, being
isomorphic to $H^{-1}_0(s_0|d)^\Delta$, constitutes a subalgebra
in the Lie algebra $H_0(s_0|d)^\Delta$. It is clear that
restricting the induced Lie bracket onto
$H_0^{-1}(s_0|d)^\Delta\simeq H_0^{-1}(\delta|d)^\Delta_0$, we
arrive at the bracket (\ref{DB}) with $\Lambda=\pi([\Omega_2])$.
Thus,  one can speak of a canonical Lie algebra isomorphism
(\ref{f}) associated to a given BRST charge.

\section{Existence and uniqueness of the BRST charge}

So far we have considered the BRST charge as given from the
outset, and the main focus has been placed on the study of the
corresponding BRST cohomology. We have shown that some of the
local BRST cohomology groups have a direct interpretation in terms
of the underlying classical dynamics.  In practice, the problem
set up is quite opposite: given the classical equations of motion,
it is necessary to construct a proper BRST charge. This problem is
similar to that in the BV quantization of Lagrangian gauge systems
where one needs to construct a local master action starting from a
gauge invariant action functional. The well known theorem \cite{H}
ensures the unique existence of the local master action under some
regularity conditions. These conditions are imposed to provide the
existence of the Koszul-Tate resolution for the stationary surface
defined by the equations of motion. The Koszul-Tate differential
appears thus the only input that comes from a classical theory
both in the Lagrangian and non-Lagrangian settings. The inductive
construction of the Koszul-Tate resolution for the stationary
surface of Lagrangian equations of motion has been proposed in
\cite{FHST}, \cite{FH}. It can be carried over immediately to
non-Lagrangian gauge theories if one starts with the extended
system of dynamical equations (\ref{ExS}) governed by the ``odd
action'' $\stackrel{_{(0)}}{\Omega}=\bar\eta^aT_a$. A
non-Lagrangian counterpart of Henneaux's theorem then reads:

\begin{theorem}
Up to a canonical transformation any classical BRST charge is
completely determined by the underlying Koszul-Tate differential,
and conversely each Koszul-Tate differential gives rise to a
classical BRST charge.
\end{theorem}

\begin{proof} Let $\delta$ be the Koszul-Tate differential associated to some  classical BRST
charge. Our strategy will be as follows. First, starting from some
$\delta$ we will construct a classical BRST charge that have the
operator  $\delta$ as its Koszul-Tate differential. The
construction is a rather straightforward application of the
homological perturbation theory \cite{HT}. Then we will show that
any two classical BRST charges with the same Koszul-Tate
differential are related to each other by a chain of  canonical
transformations. Either step will crucially exploit the acyclicity
of $\delta$ in positive resolution degree and pure ghost number.

 As in the proof of Theorem \ref{T2}
we split the fields into two classes: the ``foreground'' and the
``background'' fields. The background fields are, by definition,
the fields with zero pure ghost number and the pure ghost number
of foreground fields is strictly positive. Any local functional is
given by the integral of a power series in foreground fields and
their derivatives. This allows us to split the space of local
functionals in the direct sum  of two subspaces:
\begin{equation}\label{FF}
\mathcal{F}=\mathcal{F}'\oplus \mathcal{F}''\,.
\end{equation}
The functionals from $\mathcal{F}'$ are at most linear in
foreground fields and the functionals from $\mathcal{F}''$ are at
least quadratic in foreground fields. As is shown in Appendix A,
the general expression for the Koszul-Tate differential involves
only the background fields. Therefore,
\begin{equation}\label{dF}
\delta \mathcal{F}'=\mathcal{F}'\,,\qquad \delta
\mathcal{F}''=\mathcal{F}''\,.
\end{equation}
The subspace $\mathcal{F}''$ is also closed with respect to the
Poisson bracket, i.e.,
\begin{equation}
\{\mathcal{F}'',\mathcal{F}''\}\subset \mathcal{F}''\,.
\end{equation}

Let $\delta$ be  a Koszul-Tate operator associated to some
classical BRST charge. We are looking for a classical BRST charge
in the form
\begin{equation}\label{oo}
 \Omega_1= \Omega'_1+\Omega''_1\,,
\end{equation}
where
\begin{equation}
 \mathcal{F}'\ni\Omega_1'=\int_X
v\left(\sum_{s=1}^m\vf^{i_s}\delta\bar\vf_{i_s} +
\sum_{s=1}^{n+1}\bar\vf^{i_s}\delta\vf_{i_s}\right)\qquad
\mbox{and}\qquad \Omega_1''\in \mathcal{F}''\,.
\end{equation}
It is easy to see that
\begin{equation}\label{oF}
\{\Omega'_1,\Omega_1'\}\in \mathcal{F}''\,,\qquad
\{\Omega'_1,\,\cdot\,\}=\delta + \cdots\,.
\end{equation}
Here the dots refer to the terms with positive pure ghost number.
It follows from the second relations in (\ref{dF}) and (\ref{oF})
that $\{\Omega'_1, \mathcal{F}''\}\subset \mathcal{F}''$. Notice
that the $\Omega''_1$-part of the classical BRST charge (\ref{oo})
does not contribute to the Koszul-Tate differential as the pure
ghost number of the Hamiltonian vector field
$\{\Omega_1'',\,\cdot\,\}$ is strictly positive. So, if exists,
$\Omega_1$ gives rise to the original Koszul-Tate differential
$\delta$. The existence is easily proved by applying the machinery
of homological perturbation theory. First, we expand $\Omega_1$
according to the resolution degree,
\begin{equation}\label{rexp}
\Omega_1=\sum_{p=0}^\infty \stackrel{_{(p)}}{\Omega}{\!\!}_1=
\sum_{p=0}^\infty
(\stackrel{_{(p)}}\Omega{\!\!}'_1+\stackrel{_{(p)}}\Omega{\!\!}''_1)\,.
\end{equation}
Note that $ \deg\, \Omega'_1\leq \max(n+1,m+1)$ and $\deg\,
\Omega_1''=\pgh\,\Omega_1''\geq 2$. The master equation
$\{\Omega_1,\Omega_1\}=0$ is then equivalent to the infinite
sequence of equations
\begin{equation}\label{meqs}
\delta
\stackrel{_{(p+1)}}{\Omega_1''}=\stackrel{_{(p)}}{D}(\stackrel{_{(0)}}{\Omega}{\!\!}_1,\ldots,\stackrel{_{(p)}}{\Omega}{\!\!}_1
)\,,
\end{equation}
where
\begin{equation}\label{Dp}
\mathcal{F}''\ni\;\stackrel{_{(p)}}{D}=\frac12\sum_{s=1}^q\sum_{q=0}^p
\{\stackrel{_{(p+s-q)}}{\Omega_1},\stackrel{_{(q)}}{\Omega}_1\}^{^{(s)}}\,,\qquad\deg\,
\{\,\cdot\,,\,\cdot\,\}^{^{(s)}}=-s \,.
\end{equation}
Eqs. (\ref{meqs}) are solvable if and only if $\delta
\stackrel{_{(p)}}{D}=0$. This follows from the acyclicity of
$\delta$ in positive resolution degree and pure ghost number
(Theorem \ref{T35}). To prove the $\delta$-closedness of $
\stackrel{_{(p)}}{D}$, we can proceed by induction on $p$. For
$p=0$, the functional $\stackrel{_{(0)}}{D}$ vanishes and we have
the trivial solution $\stackrel{_{(1)}}\Omega{\!\!}_1''=0$.
Suppose that we have satisfied the first $p$ equations of sequence
(\ref{meqs}) and let
$$
R_p=\sum_{s=0}^p\stackrel{_{(s)}}{\Omega}_1
$$
be the sum of already known functionals
$\stackrel{_{(s)}}{\Omega}_1$. By induction hypothesis
$$
\{R_p,R_p\}=2\stackrel{_{(p)}}{D}+\cdots\,,
$$
where the unwritten terms are of resolution degree $>p$.
Extracting the leading term of the Jacobi identity
$\{R_p,\{R_p,R_p\}\}=0$, we find that  $\delta
\stackrel{_{(p)}}{D}=0$. Thus, Eq. (\ref{meqs}) possesses a
solution and, by virtue of  (\ref{dF}), we can choose
$\stackrel{_{(p+1)}}{\Omega''_1}$ to be an element of
$\mathcal{F}''$.

From the consideration above it follows that the  first two terms
$\stackrel{_{(0)}}{\Omega}_1$ and $\stackrel{_{(1)}}{\Omega}_1$ of
the classical BRST charge (\ref{rexp}) are unambiguously
determined by the Koszul-Tate differential $\delta$. By contrast,
the higher order terms $\stackrel{_{(p)}}{\Omega}_1$, $p\geq 2$,
are not so unique as Eq. (\ref{meqs}) defines them only up to a
$\delta$-coboundary. Let us show that this ambiguity can be
entirely absorbed by an appropriate canonical transformation of
$\mathcal{F}$. Consider two different classical BRST charges
$\Omega_1$ and ${\overline{\Omega}}_1$ that share one and the same
Koszul-Tate differential $\delta$ and coincide to each other up to
the $k$-th order in resolution degree. In view of the comment
above, this means $k\geq 1$. By (\ref{meqs}) the difference
$\stackrel{_{(k+1)}}{{\overline{\Omega}}_1}
-\stackrel{_{(k+1)}}{{\Omega}_1}$ is $\delta$-closed. Since both
the resolution degree and pure ghost number of the difference are
positive, we can write
$$
\stackrel{_{(k+1)}}{{{\overline{\Omega}}}_1} =
\stackrel{_{(k+1)}}{{{\Omega}}_1}+\delta {\Phi}
$$
for some ${\Phi}\in \mathcal{F}''$ of resolution degree $k+2$. The
functional $\Phi$ generates a well-defined canonical
transformation $e^{\{\Phi,\,\cdot\,\}}: \mathcal{F}\rightarrow
\mathcal{F}$. (Since $\Phi\in \mathcal{F}''$, only finitely many
terms of the exponential series
$e^{\{H,\,\cdot\,\}}=1+\{H,\,\cdot\,\}+\cdots$ contribute to each
order in pure ghost number.) Applying this transformation to
$\overline{\Omega}_1$, we get a new solution to the master
equation
$$
\widetilde{\Omega}_1=e^{\{{\Phi},\,\cdot\,\}}\overline{{\Omega}}_1=\overline{\Omega}_1-\delta
\Phi+\cdots\,.
$$
By definition, $\widetilde{\Omega}_1$ coincides with $\Omega_1$ at
least to order $k+1$ and gives rise to the same Koszul-Tate
operator. So, the classical BRST charges $\Omega_1$ and
$\overline{\Omega}_1$ appear to be canonically equivalent to the
$(k+1)$-st order. The same arguments, being applied to the pair
$\Omega_1$, ${\widetilde{\Omega}}_1$, show that $\Omega_1$ is
canonically equivalent to $\widetilde{\Omega}_1$, and hence to
$\overline{\Omega}_1$, up to order $k+2$. Repeating these
arguments once and again, we infer the canonical equivalence of
$\Omega_1$ and $\overline{\Omega}_1$ at all orders in resolution
degree.

\end{proof}

Though useful (e.g. for defining the conservation laws and rigid
symmetries), the classical BRST charge carries no valuable
information about quantum properties of the system.  The quantum
properties are determined by the higher order terms in the
expansion of the ``total'' BRST charge (\ref{Om-exp}) with respect
to the momentum degree. Letting all the higher terms to be zero,
which is admissible, and applying  the general quantization
procedure \cite{KazLS} to the BRST charge $\Omega=\Omega_1$, one
can see that the probability amplitude on the space of fields is
given by the singular distribution supported on the solutions to
the classical equations of motion. As a result, the quantum
averages of physical observables are reduced to their classical
values with no quantum fluctuations. Thus, to get a nontrivial
quantization, one has to admit higher order terms in the BRST
charge. These can be considered as a deformation of the classical
BRST charge by higher order terms in momentum degree. In the rest
of the section, we are going to clarify the structure of such
deformations by making use of previously obtained  knowledge of
the local BRST cohomology. To this end, consider the decomposition
of the BRST charge into the sum
\begin{equation}\label{qbrst}
    \Omega=\Omega_1+\mathbf{\Omega}\,,\qquad
    \mathbf{\Omega}=\sum_{s=2}^\infty\Omega_s\,,\qquad \Deg\,
    \Omega_s =s\,.
\end{equation}
It is the functional  $\mathbf{\Omega}$ that is regarded  as a
deformation of the classical BRST charge $\Omega_1$. Given
$\Omega_1$, the master equation $\{\Omega,\Omega\}=0$ takes the
form of the Maurer-Cartan equation
\begin{equation}\label{MC}
    s_0\mathbf{\Omega}=\{\mathbf{\Omega},\mathbf{\Omega}\}
\end{equation}
with respect to the classical BRST differential
\begin{equation}
s_0=\{\Omega_1,\,\cdot\,\}=\delta+\cdots \,,\qquad s_0^2=0\,.
\end{equation}
The deformation  $\mathbf{\Omega}$ is said to be of order $p$ if
its expansion in homogeneous components  starts with a term of
momentum degree $p$,
$$
\mathbf{\Omega}=\sum_{m=p}^\infty{\Omega}_m\,.
$$
Two deformations $\mathbf{\Omega}$ and $\mathbf{\Omega}'$ are said
to be equivalent if they are related to each other by a canonical
transformation, i.e., there exists a local functional $\Phi$ with
$\mathrm{Deg}\,\Phi\geq 2$ such that
\begin{equation}\label{cantr}
e^{\{\Phi,\,\cdot\,\}}(\Omega_1+\mathbf{\Omega})=\Omega_1+\mathbf{\Omega}'\,.
\end{equation}
A \textit{trivial deformation} is a deformation that is equivalent
to zero. We say that a deformation of order $p$ is \textit{strict}
if it is not equivalent to a deformation of order $>p$. Let
$\mathrm{Def}(\delta)=\bigcup_{m\geq 2}\mathrm{Def}_m(\delta)$
denote the set of all solution to the Maurer-Cartan equation
(\ref{MC}) modulo canonical transformations; here the subset
$\mathrm{Def}_m(\delta)$ consists of the strict deformations of
order $m$. The deformations of $\mathrm{Def}_2(\delta)$ are called
\textit{regular} and the deformations of order $>2$ are called
\textit{singular}.

On substituting the expansion (\ref{qbrst}) into (\ref{MC}), we
arrive at the sequence of equations
\begin{equation}\label{def}
\begin{array}{ll}
s_0{\Omega}_m=0\,, & m=p,\ldots, 2p-2\,;\\[3mm]\displaystyle
s_0{\Omega}_m=-\frac 12
\sum_{n=p}^{m-2}\{{\Omega}_{m-n},{\Omega}_{n+1}\}\,,\qquad&
m>2p-2\,.
\end{array}
\end{equation}
The first equation of the sequence, $s_0\Omega_p=0$, means that
the leading term of the deformation is a $s_0$-cocycle. For a
strict deformation this cocycle is necessarily nontrivial for if
$\mathbf{\Omega}=s_0\Phi$, then the canonical transformation
(\ref{cantr}) gives the BRST charge $\Omega'$ with $\mathrm{Deg}\,
\mathbf{\Omega}'
>p$.   We have the map
$$
\mathrm{Def}_m(\delta)\rightarrow H^1_m(s_0|d)^\Delta
$$
associating to each deformation its leading term. In general this
map is neither injective nor surjective. The latter fact implies
that not all $s_0$-closed functionals of ghost number 1 and
momentum degree $p\geq 2$ can serve as leading terms of $p$-th
order deformations. Consider for example a regular deformation,
that is, a strict deformation of order 2. It follows from
(\ref{def}) that at the third order in momentum degree the leading
term of such a deformation, $\Omega_2$, obeys the equation
\begin{equation}\label{om3}
s_0\Omega_3=\frac12\{\Omega_2,\Omega_2\}\,.
\end{equation}
By the Jacobi identity the right hand side of this equation is
$s_0$-closed whenever $\Omega_2$ is a $s_0$-cocycle. But the
equation requires the Poisson square of $\Omega_2$ to be
$s_0$-exact. In other words, the Maurer-Cartan equation possesses
a solution at the third order in momentum degree iff
$$\{[\Omega_2],[\Omega_2]\}=[\{\Omega_2,\Omega_2\}]=0\in H^2_3(s_0|d)^\Delta\,.$$
At the next step  we get the equation
$$
D\Omega_4=\{\Omega_3,\Omega_2\}\,.
$$
Again, the right hand side of this equation is $s_0$-closed but not
generally $s_0$-exact.  The class $[\{\Omega_3,\Omega_2\}]$, where
$\Omega_3$ is defined by Eq. (\ref{om3}), is known as the Massey
cube of the class $[\Omega_2]$. It is usually denoted by
$\langle[\Omega_2],[\Omega_2],[\Omega_2]\rangle$. Vanishing of the
class $\langle[\Omega_2],[\Omega_2],[\Omega_2]\rangle$ is thus the
necessary and sufficient condition for the Maurer-Cartan equation
(\ref{MC}) to admit a solution at the fourth order in momentum
degree. It should be emphasized that the Massey cube can only be
defined for a class $[\Omega_2]$ with vanishing ``Massey square''
$\langle[\Omega_2], [\Omega_2]\rangle=\{[\Omega_2],[\Omega_2]\}$ and
Eq. (\ref{om3}) defines $\Omega_3$ up to an arbitrary $s_0$-cocycle.
So, one can regard the assignment $[\Omega_2]\mapsto
\langle[\Omega_2],[\Omega_2],[\Omega_2]\rangle$ as a partially
defined and multivalued map from $H_3^2(s_0|d)^\Delta$ to
$H_4^2(s_0|d)^\Delta$. Proceeding in the same manner one can see
that the class $[\Omega_2]\in H_2^1(s_0|d)^\Delta$ extends to an
element of $\mathrm{Def}_2(\delta)$ iff all the Massey powers of
$[\Omega_2]$ can be made zero,
\begin{equation}\label{MBr}
\langle \underbrace{[\Omega_2],[\Omega_2],\ldots,
[\Omega_2]}_{n}\rangle\ni 0\qquad \forall n=2,3,\ldots \,.
\end{equation}
For a  general definition  of the Massey products in the category
of differential graded Lie algebras we refer the reader to
\cite{R}, \cite{FL}.

A similar analysis applies to singular deformations. In
particular, the first equation in (\ref{def}) shows that the
square $\{[\Omega_p],[\Omega_p]\}\in H^2_{2p-1}(s_0|d)^\Delta$ of
the leading term of a $p$-th order deformation has to vanish. This
ensures the existence of a solution to (\ref{MC}) at order $2p-1$.
The study of higher order obstructions to solvability of the
Maurer-Cartan equation appears to be much  more involved than in
the case of regular deformations, but it can still  be formulated
in terms of Massey-like products \cite{FL}. We will not go into
details here.

Turning back to the regular deformations, we see that the existence
of a solution to the Maurer-Cartan equation (\ref{MC}) is generally
obstructed by elements of the cohomology groups
$H^2_m(s_0|d)^\Delta$, $m>2$. By Theorem \ref{p}  all these groups
are isomorphic to the corresponding groups
$H^2_m(\delta|d)^\Delta_0$. Furthermore, the  isomorphism (\ref{pi})
can be regarded as an isomorphism of graded Lie algebras  if we
endow $H^2_m(\delta|d)^\Delta_0$ with the pulled-back Lie bracket
given by the r.h.s. of (\ref{LB})\footnote{It should be noted that
the pulled-back Lie bracket does not coincide with the naive Poisson
bracket on $H^2_m(\delta|d)^\Delta_0$ unless the resolution degrees
of multiplied classes are less than or equal to 1.}. This allows us
to rewrite the condition (\ref{MBr}) in the following form:
\begin{equation}\label{mp}
H^2_{n+1}(\delta|d)^\Delta_0\supset \langle
\underbrace{\Lambda,\Lambda,\ldots, \Lambda}_{n}\rangle\ni 0\qquad
\forall n=2,3,\ldots \,,
\end{equation}
 $\Lambda=\pi ([\Omega_2])$ being the Lagrange structure underlying
the BRST charge $\Omega$. Thus, we conclude that a Lagrange
structure $\Lambda$ gives rise to a regular deformation of a
classical BRST charge iff all its Massey powers can be made zero.
In particular, vanishing of the Massey square of $\Lambda$ amounts
to its integrability by Definition (\ref{int}).

Actually, only a finite number of classes of the sequence
(\ref{mp}) may present true obstacles to the existence of
$\Omega$, as for all $n>\Delta+1$,
$H_{n+1}^2(\delta|d)_0^\Delta=0$  by Theorem \ref{T32}. For
example, if $\Delta=1$ (the case of mechanical systems) each
integrable Lagrange structure gives rise to a regular deformation,
and conversely each nontrivial deformation is necessarily regular
and is defined by some integrable Lagrange structure. In our next
paper we will show that there is a one-to-one correspondence
between the integrable Lagrange structures in one dimension and
the weak Poisson structures introduced in  \cite{LS}. This
correspondence allows one to relate the path-integral quantization
of not necessarily Lagrangian mechanical systems with their
deformation quantization. For actual field theories ($\Delta
> 0$) singular deformations may exist as well. The above analysis
shows, for instance,  that any element of
$H_p^1(\delta|d)^\Delta_0$ with $p>(\Delta+3)/{2}$ generates  a
singular deformation of order $p$. The quantization by means of
singular deformations, having no analogue in Lagrangian field
theory, looks  rather intriguing, to say the least. Its
interpretation in terms of deformation quantization remains
obscure to us, though.

We summarize the main results of this section by the following

\begin{theorem}
Given a Koszul-Tate differential $\delta$ and a Lagrange structure
$\Lambda\in H^1_2(\delta|d)^\Delta_0$ obeying the finite set of
conditions
$$
\langle \underbrace{\Lambda,\Lambda,\ldots ,
\Lambda}_{n}\rangle\ni 0\qquad \forall n=2,3,\ldots, \Delta+1\,,
$$
there exists a local BRST charge $\Omega=\sum_{k=1}^\infty
\Omega_k$ such that
$$
\{\Omega_1,\,\cdot\,\}=\delta+ (\mbox{terms of resolution degree
$\geq 0$}) \quad \mbox{and}\quad \pi([\Omega_2])=\Lambda\,.
$$
In general, the $\Omega$ is not unique even when considered up to
a canonical equivalence.
\end{theorem}

\section{Conclusion}

In this paper, we have extended the main theorems on the local
BRST cohomology beyond the scope of Lagrangian dynamics. Let us
comment on the most curious peculiarities of the non-Lagrangian
BRST complex and discuss further perspectives.

In the first place, we see that the groups
$H_1^{0}(\delta|d)^\Delta_0$ and $H_0^{-1}(\delta|d)^\Delta_0$,
whose elements are identified with nontrivial rigid symmetries and
characteristics, are not generally isomorphic to each other unlike
it happens in Lagrangian dynamics. These groups, however, combine
within the unifying group $H^{(1)}(\delta|d)_0^\Delta=\bigoplus_g
H^g_{g+1}(\delta|d)_0^\Delta$. The group
$H^{(1)}(\delta|d)_0^\Delta$ is shown to have the structure of a
graded Lie algebra of depth -1. It then follows that the
homogeneous subgroup $H_1^{0}(\delta|d)^\Delta_0$ of ghost number
zero is a subalgebra in $H^{(1)}(\delta|d)_0^\Delta$ (isomorphic
to the Lie algebra of rigid symmetries) and the bottom subspace
$H_0^{-1}(\delta|d)^\Delta_0$ is a module over the Lie algebra
$H_1^{0}(\delta|d)^\Delta_0$ with respect to the adjoint action
(the symmetries act on the characteristics). The next homogeneous
subspace of the unifying group is $H_2^{1}(\delta|d)^\Delta_0$. It
is understood as the space of Lagrange structures. Again, by a
purely algebraic reason, the Lie bracket of a Lagrange structure
and a characteristic is a rigid symmetry. Thus, each Lagrange
structure defines a  map from the space of characteristics to the
space of rigid symmetries. In the Lagrangian case, the space
$H_2^{1}(\delta|d)^\Delta_0$ contains a distinguished element -
the canonical Lagrange structure - for which the aforementioned
map is injective and covers all the symmetries  of the
corresponding action functional. This is the content of the first
Noether's theorem. For an arbitrary Lagrange structure, the map is
neither injective nor surjective, so that the symmetries and
conservation laws cannot be so tightly related to each other
beyond the class of Lagrangian field theories.

Besides being a module over the Lie algebra of rigid symmetries,
the space of conservation laws of any Lagrangian field theory
carries its own Lie algebra structure with respect to the Dickey
bracket.  Like the Noether isomorphism, the Dickey bracket owes
its existence to the canonical Lagrange structure. When trying to
extend the Dickey bracket beyond the Lagrangian setting, one
naturally arrives at the concept of integrability of Lagrange
structure. Namely, a Lagrange structure $\Lambda$ is said to be
integrable if it satisfies the Maurer-Cartan equation $
[\Lambda,\Lambda]=0 $. Given an integrable Lagrange structure, one
can define  a non-Lagrangian counterpart of the Dickey bracket as
the derived bracket of characteristics (\ref{DB}) and the same
Lagrange structure defines the Lie algebra homomorphism (\ref{f}).
The notion of a Lagrange structure appears thus to be a principal
connecting-link for the groups of symmetries and conservation
laws. The interpretation of the higher homogeneous subgroups of
$H^{(1)}(\delta|d)_0^\Delta$ remains  unclear at present, maybe
excluding the group $H_3^{2}(\delta|d)^\Delta_0$ that can be
regarded as the obstruction space for integrability of Lagrange
structures. Further studies of the local BRST cohomology will
hopefully shed more light on the   utility of these groups.

The study of symmetries and conservation laws of  partial
differential equations is a broad area of mathematical physics with
numerous applications and well developed methodology. It seems that
similar methods apply to computation of the group
$H_2^{1}(\delta|d)^\Delta_0$ of Lagrange structures, though it can
be a difficult problem, in general.  In any case, every single
Lagrange structure gives a valuable bit of information about
classical dynamics and, what is important, defines a reasonable
quantum theory. As far as the quantization problem is concerned, one
is usually interested in Lagrange structures subject to one or
another set of physical conditions, rather than in knowing the
entire space $H^1_2(\delta|d)^\Delta_0$. For example, if some
fundamental symmetries of classical dynamics are expected to survive
at the quantum level, the same invariance conditions should  be
imposed on the Lagrange structure. (Notice that the space of
Lagrange structures $H^1_2(\delta|d)^\Delta_0$ is a module over the
Lie algebra of rigid symmetries $H^0_1(\delta|d)^\Delta_0$). These
conditions together with some other physical requirements may
strongly restrict the possible choice of a Lagrange structure or
even make it unique. An additional set of conditions comes from
requiring the existence of a local BRST charge. Unlike the usual BV
formalism for Lagrangian gauge theories, the locality of equations
of motion and a compatible Lagrange structure $\Lambda$ does not
ensure the existence of a local BRST charge. In order for such a
charge to exist, it is necessary and sufficient that all the Massey
powers $\langle\Lambda,\Lambda,\ldots,\Lambda\rangle$ of the
cohomology class $\Lambda\in H^1_2(\delta|d)^\Delta_0$ to vanish. In
particular, the vanishing of the Massey square reproduces the
integrability condition.

Finally, there can exist higher-order deformations of the
classical BRST charge that are governed by elements of the groups
$H^1_m(\delta|d)^\Delta_0$ with $m>2$. These deformations, called
singular, are unrelated to any Lagrange structure and their
relevance to the path-integral quantization of non-Lagrangian
field theories invites further studies.

\vspace{5mm} \noindent
 {\textbf{Acknowledgments}.}
 We dedicate this paper to the 70th birthday of Igor Victorovich
 Tyutin, who contributed so much to theory of gauge systems,
 and who so much helped many of his colleagues, in many ways, during many years.

The work is partially supported by the Federal Targeted Program
under the state contract P1337 and by the RFBR grant
09-02-00723-a. SLL acknowledges the support from the RFBR grant
11-01-00830. We are thankful  to Elena Mosman for her valuable
comments on the first version of the manuscript.

\appendix\label{A1}
\section{The proofs of two auxiliary statements}

Any Hamiltonian vector field $X_H=\{H,\,\cdot\,\}$ generated by a
local functional $H\in \mathcal{F}$ is decomposed into the sum
$$
X_H=V_H+U_H
$$
of two variational vector fields\footnote{All the variational
derivatives act on the left.}
\begin{equation}\label{VU}
\begin{array}{l}
\displaystyle V_H=\sum_{s} \int_X v
\left((-1)^{s(\varepsilon(H)-1)-1}\frac{\delta H}{\delta
\bar\vf{}^{i_s}}\frac{\delta}{\delta
\vf_{i_s}}+(-1)^{(s-1)\varepsilon(H)}\frac{\delta H}{\delta
\vf^{i_{s-1}}}\frac{\delta}{\delta
\bar\vf{}_{i_{s-1}}}\right)\,,\\[6mm]
\displaystyle U_H=\sum_{s}\int_X
v\left((-1)^{s\varepsilon(H)}\frac{\delta H}{\delta
\vf_{i_s}}\frac{\delta}{\delta
\bar\vf{}^{i_s}}+(-1)^{s(\varepsilon(H)-1)+\varepsilon(H)}\frac{\delta
H}{\delta\bar\vf{}_{i_{s-1}}}\frac{\delta}{\delta
\vf^{i_{s-1}}}\right)
\end{array}
\end{equation}
such that
\begin{equation}\label{UV}
\pgh \,V_H\geq 0\,,\qquad \deg \,U_H\geq 0\,.
\end{equation}

To prove Proposition \ref{p1} consider the Hamiltonian vector
field $X_\Omega=V_\Omega+U_\Omega$ generated by the BRST charge
$\Omega$. As the resolution degree of $U_\Omega$ is bounded from
below by zero, it remains to evaluate the low bound for
$V_\Omega$. Applying the identity (\ref{N}) yields
\begin{equation}\label{VV}
    \deg\, V_\Omega=\pgh\, V_\Omega+ \mathrm{Deg}\, V_\Omega-\gh\, V_\Omega\,.
\end{equation}
On the other hand,
\begin{equation}\label{VVV}
\gh\, V_\Omega=1\,,\qquad \mathrm{Deg}\,
V_\Omega=\mathrm{Deg}\,\Omega -1\geq 0\,.
\end{equation}
Combining (\ref{UV}), (\ref{VV}), and (\ref{VVV}) we get
\begin{equation}\label{}
    \deg\,V_\Omega =\pgh\, V_\Omega+\mathrm{Deg}\, \Omega-2\geq -1\,.
\end{equation}
The last inequality shows two things: (i) only the classical BRST
charge $\Omega_1$ contributes to $\delta$ and (ii) $\pgh \,\delta
=0$. From the latter property and the  definition of $V_\Omega$
it also follows that the Koszul-Tate differential $\delta$ is
completely determined in terms of fields of pure ghost number zero
(background fields).

Similar arguments allows one to prove the equality (\ref{LB}). Let
$a$ and $b$ be relative $\delta$-cocycles representing the classes
$\alpha, \beta \in H^{(1)}(\delta|d)^{\Delta}_0$. By the
definition of the isomorphism (\ref{pi}), the classes
$\pi^{-1}(\alpha)$ and $\pi^{-1}(\beta)$ are represented by
relative $s_0$-cocycles of the form $a+A$ and $b+B$ with $\pgh\,
A>0$ and $\pgh\,B>0$. To prove the equality in question we only
need to show that the resolution-degree-one part of the Poisson
bracket
\begin{equation}\label{PPB}
\{a+A,b+B\}=\{a,b\}+\{a,B\}+\{A,b+B\}
\end{equation}
is given by $\{a,b\}$. The equality $\deg\, \{a,b\}=1$ is obvious
as $a$ and $b$ belong to the Lie algebra $\mathcal{Z}$, see Sec.
\ref{interp}. What is left is to show  that the resolution degree
of the other two terms in (\ref{PPB}) is greater than 1. Applying
(\ref{VU}) we can write $\{a, B\}=V_aB+U_aB$, where
\begin{equation}\label{ineq}
\pgh\, V_aB\geq \pgh\, B>0\,,\qquad \deg\,U_aB\geq \deg\,B>0\,.
\end{equation}
On the other hand,  for any  $[c]\in H^g_{g+1}(s_0|d)^\Delta$ we
have
$$\deg\, c=\pgh\, c+1\,.$$
In particular, this equality is true for (each term of) the
Poisson bracket (\ref{PPB}). Combining the last fact with the
second inequality (\ref{ineq}) we get
$$
\pgh\,U_aB\geq \pgh\, B>0\,.
$$
Together with the first inequality (\ref{ineq}) this yields
$\pgh\, \{a,B\}\geq \pgh \, B
> 0$. In the same way one can verify  that $\pgh\, \{A,b+B\}\geq
\pgh\, A>0$.

\end{document}